\documentclass[a4paper,11pt]{article}
\usepackage[T1]{fontenc}

\usepackage{a4wide}
\usepackage{graphicx}
\usepackage{esvect} 
\usepackage{tcolorbox}
\usepackage{complexity}
\usepackage{tikz}
\usepackage{amsmath}
\usepackage{mathrsfs}
\usepackage{multibib}
\usepackage{enumerate}
\usepackage{enumitem}
\usepackage{algorithm}
\usepackage{algorithmic}
\usepackage{booktabs}
\usepackage{multirow}
\usepackage{lineno}
\usepackage{tabularx}
\usepackage{arydshln}
\usepackage{framed} 
\usepackage{amsthm,thmtools}
\usepackage{authblk}
\usepackage{hyperref}
\usepackage{cleveref}

\usetikzlibrary{calc,positioning,shapes}
\usetikzlibrary{decorations.markings}
\usetikzlibrary{arrows.meta}

\declaretheorem[style=plain,name=Theorem]{theorem}
\declaretheorem[style=plain,sibling=theorem,name=Lemma]{lemma}

\crefname{theorem}{Theorem}{Theorems}
\crefname{proposition}{Proposition}{Propositions}
\crefname{lemma}{Lemma}{Lemmas}
\crefname{exmp}{Example}{Examples}
\crefname{corollary}{Corollary}{Corollarys}
\crefname{claim}{Claim}{Claims}
\crefname{remark}{Remark}{Remarks}
\crefname{section}{Section}{Sections}

\newcommand{\prb}[1]{\textnormal{\scshape #1}}
\newcommand{\gad}[1]{\textnormal{\scshape #1}}

\newcommand{\Rule}{\mathsf{R}}
\newcommand{\kTJ}{k\text{-}\mathsf{TJ}}

\newcommand{\TJ}{\mathsf{TJ}}
\newcommand{\kTS}{k\text{-}\mathsf{TS}}

\newcommand{\TS}{\mathsf{TS}}
\newcommand{\TAR}{\mathsf{TAR}}
\newcommand{\True}{T}
\newcommand{\False}{F}

\newcommand{\NCL}{\prb{Nondeterministic Constraint Logic}}

\newcommand{\symdif}{\mathbin{\triangle}}

\newcommand{\notinA}{\bar{A}}

\theoremstyle{plain}
\newtheorem{Introtheorem}{Theorem}

\newcommand{\rev}[1]{#1}

\newcommand{\revcr}[1]{#1}

\begin{document}
\title{Reachability of Independent Sets and Vertex Covers\\ Under Extended Reconfiguration Rules}

\author[1]{Shuichi Hirahara
\thanks{s\_hirahara@nii.ac.jp}}
\author[2]{Naoto Ohsaka
\thanks{ohsaka\_naoto@cyberagent.co.jp}}
\author[3]{Tatsuhiro Suga
\thanks{suga.tatsuhiro.p5@dc.tohoku.ac.jp}}
\author[4]{Akira Suzuki
\thanks{akira@tohoku.ac.jp}}
\author[3]{\\Yuma Tamura
\thanks{tamura@tohoku.ac.jp}}
\author[3]{Xiao Zhou
\thanks{zhou@tohoku.ac.jp}}

\affil[1]{National Institute of Informatics, Tokyo, Japan}
\affil[2]{CyberAgent, Inc., Tokyo, Japan}
\affil[3]{Graduate School of Information Sciences, Tohoku University, Sendai, Japan}
\affil[4]{Center for Data-Driven Science and Artificial Intelligence, Tohoku University, Sendai Japan}

\date{}
\maketitle          
\begin{abstract}
    In reconfiguration problems, we are given two feasible solutions to a graph problem and asked whether one can be transformed into the other via a sequence of feasible intermediate solutions under a given reconfiguration rule. While earlier work focused on modifying a single element at a time, recent studies have started examining how different rules impact computational complexity.

    Motivated by recent progress, we study \prb{Independent Set Reconfiguration} (\prb{ISR}) and \prb{Vertex Cover Reconfiguration} (\prb{VCR}) under the $k$-Token Jumping ($\kTJ$) and $k$-Token Sliding ($\kTS$) models. In $\kTJ$, up to $k$ vertices may be replaced, while $\kTS$ additionally requires a perfect matching between removed and added vertices. It is known that the complexity of \prb{ISR} crucially depends on $k$, ranging from $\PSPACE$-complete and $\NP$-complete to polynomial-time solvable.
    
    In this paper, we further explore the gradient of computational complexity of the problems.
    We first show that \prb{ISR} under $\kTJ$ with $k = |I| - \mu$ remains $\NP$-hard when $\mu$ is any fixed positive integer and the input graph is restricted to graphs of maximum degree 3 or planar graphs of maximum degree 4, where $|I|$ is the size of feasible solutions. In addition, we prove that the problem belongs to $\NP$ not only for $\mu=O(1)$ but also for $\mu = O(\log |I|)$. In contrast, we show that \prb{VCR} under $\kTJ$ is in $\XP$ when parameterized by $\mu = |S| - k$, where $|S|$ is the size of feasible solutions. Furthermore, we establish the $\PSPACE$-completeness of \prb{ISR} and \prb{VCR} under both $\kTJ$ and $\kTS$ on several graph classes, for fixed $k$ as well as superconstant $k$ relative to the size of feasible solutions.
\end{abstract}


\section{Introduction}\label{sec:Introduction}
Reconfiguration problems ask whether it is possible to reach a target state from an initial state by gradually transforming one feasible solution into another, where each intermediate solution must also be feasible.  
At each step, the current solution can be changed to an ``adjacent'' one, as defined by a given restriction known as a \emph{reconfiguration rule}.  
One of the most well-known examples is the \emph{15-puzzle}, where the rule permits sliding a tile into an adjacent empty space.
Reconfiguration problems have been extensively studied in the context of classical graph problems involving feasible solutions such as independent sets~\cite{BKLMOS21,BB17,BonsmaKW14,BFHM21,DDFHIOOUY15,HD05,IDHPSUU11,KMM12,LM19}, cliques~\cite{DBLP:journals/dam/ItoOO23}, vertex covers~\cite{DBLP:journals/algorithms/MouawadNRS18}, dominating sets~\cite{DBLP:journals/tcs/HaddadanIMNOST16, DBLP:journals/jco/SuzukiMN16}, and so on (see surveys~\cite{DBLP:journals/algorithms/Nishimura18,DBLP:books/cu/p/Heuvel13}).
Three standard reconfiguration models have been widely studied in the literature: 
\emph{token jumping}~($\TJ$)~\cite{IDHPSUU11,KMM12}, \emph{token sliding}~($\TS$)~\cite{BKLMOS21,BonsmaKW14,HD05, IDHPSUU11}, and \emph{token addition/removal}~($\TAR$)~\cite{IDHPSUU11,KMM12} models.
In the $\TJ$ model, one can simultaneously remove any vertex from the current solution and add any vertex outside it.
The $\TS$ model is a restricted version of $\TJ$, where the removed  vertex and the added vertex must be adjacent.  
In the $\TAR$ model, vertices can be added or removed as long as the resulting set remains above (or below) a specified size threshold.

Reconfiguration problems on graphs under those rules have attracted attention in theoretical computer science, and the computational complexity of such problems has been settled~\cite{IDHPSUU11,KMM12,DBLP:journals/algorithmica/MouawadN0SS17,W18}.
Besides, the field of reconfiguration problems is in the course of trying to apply the theoretical viewpoint to functional implementation for practical use~\cite{DBLP:conf/ecai/ChristenE0MPPSS23,DBLP:conf/cpaior/ItoKNSSTT23,DBLP:conf/ictai/TodaIKSST23,DBLP:conf/walcom/YamadaBISU24}.

However, one may find that, in practical scenarios, conventional rules such as changing one element at a time may be too restrictive.  
For example, reconfiguring a monitoring system or an infrastructure network often requires multiple simultaneous changes due to physical or operational constraints.  
Even if the system cannot be reconfigured under standard one-element rules, in practice, it is unjustified to conclude that it is ``unreachable.''
More flexible reconfiguration processes---such as those allowing multiple simultaneous changes---more accurately reflect the feasibility of real-world systems.

Motivated by recent progress and aiming to address the emerging issues, several studies have begun to analyze the computational complexity of reconfiguration problems under extended reconfiguration rules~\cite{DBLP:journals/dam/BergJM18,DBLP:conf/walcom/DomonST024,DBLP:BoundedHopToken,DBLP:GeneralizedTokenJumping,ISIsoRsuga25}.

We study the reconfiguration problem under the extended rules, $k$-Token Jumping ($\kTJ$) and $k$-Token Sliding ($\kTS$)~\cite{DBLP:journals/dam/BergJM18,ISIsoRsuga25}, which allow the simultaneous exchange of up to $k$ vertices.

\subsection{Our problems}
In this paper, all graphs are simple and undirected.
\revcr{For two sets $A$ and $B$, the set difference $A\setminus B$ is $\{x\in A\colon x\notin B\}$, and $A\bigtriangleup B$ denotes the symmetric difference between $A$ and $B$, that is, $(A\setminus B) \cup (B\setminus A)$.}
For two vertex subsets $A$ and $B$ of a graph $G$ with $|A|=|B|$, we say that they are \emph{adjacent} under \emph{$k$-Token Jumping} ($\kTJ$) if $|A \symdif B| \leq 2k$.
On the other hand, they are said to be \emph{adjacent} under \emph{$k$-Token Sliding} ($\kTS$) if $|A \symdif B| \leq 2k$ and there exists a perfect matching between $A \setminus B$ and $B \setminus A$ in the graph $G$.
Intuitively, the transformation from $A$ to $B$ can be seen as moving tokens placed on the vertices in the symmetric difference $A \symdif B$.
Under $\kTJ$, up to $k$ tokens can be moved simultaneously to any vertices in $G$.
In contrast, under $\kTS$, up to $k$ tokens can be moved simultaneously along the edges of $G$.
Note that $1$-$\TJ$ and $1$-$\TS$ coincide with $\TJ$ and $\TS$, respectively.

\prb{Independent Set} and \prb{Vertex Cover} are $\NP$-complete problems~\cite{GJ79} that are thoroughly explored in graph theory.
We define their reconfiguration variants, that is, \prb{Independent Set Reconfiguration} (\prb{ISR}) and \prb{Vertex Cover Reconfiguration} (\prb{VCR}).
Recall that a \rev{set} of vertices in a graph \( G \) is called an \emph{independent set} if no two vertices in the set are adjacent in \( G \), and a \emph{vertex cover} if every edge in \( G \) has at least one endpoint in the \rev{set}.

In the \prb{Independent Set Reconfiguration} problem, we are given a graph $G$, two independent sets $I$ and $J$ of $G$ (representing the ``initial'' and ``target'' configurations of tokens, respectively) such that $|I| = |J|$, and a reconfiguration rule $\Rule \in \{\kTJ, \kTS\}$.
Then, the problem asks whether there exists a sequence $\sigma=\langle I=I_0,I_1,\ldots,I_{\ell}=J \rangle$ of independent sets of $G$, where any two consecutive independent sets in $\sigma$ are adjacent under $\Rule$.
Similarly, in \prb{Vertex Cover Reconfiguration}, we are given a graph $G$, two vertex covers $S$ and $T$ of $G$ (representing the ``initial'' and ``target'' configurations of tokens, respectively) such that $|S|=|T|$, and a reconfiguration rule $\Rule\in\{\kTJ,\kTS\}$. Then, the problem asks whether there exists a sequence $\sigma=\langle S=S_0,S_1,\ldots,S_{\ell}=T \rangle$ of vertex covers of $G$, where any two consecutive vertex covers in $\sigma$ are adjacent under $\Rule$.
In each problem, we refer to such a sequence of independent sets or vertex covers as a \emph{reconfiguration sequence}, where $\ell$ is the length of the reconfiguration sequence.
Furthermore, we say that two vertex sets are \emph{reconfigurable} under $\Rule$ if there exists a reconfiguration sequence between them under $\Rule$.
Formally, the two problems are defined as follows.

\begin{description}[leftmargin=1em, nosep, topsep=1em]
    \item[Problem] \prb{Independent Set Reconfiguration (ISR)}
    \item[Input] A simple undirected graph $G$, two independent sets $I$ and $J$ of $G$ with the same size, and a reconfiguration rule $\R\in\{\kTJ,\kTS\}$.
    \item[Output] Are $I$ and $J$ reconfigurable under $\R$?
\end{description}

\begin{description}[leftmargin=1em, nosep]
    \item[Problem] \prb{Vertex Cover Reconfiguration (VCR)}
    \item[Input] A simple undirected graph $G$, two vertex covers $S$ and $T$ of $G$ with the same size, and a reconfiguration rule $\R\in\{\kTJ,\kTS\}$.
    \item[Output] Are $S$ and $T$ reconfigurable under $\R$?
\end{description}

\paragraph{Related work.}
\prb{ISR} under $\TJ$ and $\TS$ is $\PSPACE$-complete even for planar graphs of maximum degree 3 and bounded bandwidth~\cite{HD05,DBLP:conf/iwpec/Zanden15, W18}, and perfect graphs~\cite{KMM12}. Under $\TJ$, it is known that any two independent sets of an even-hole-free graph are reconfigurable~\cite{KMM12}. Under $\TS$, the problem remains $\PSPACE$-complete for split graphs~\cite{BKLMOS21}, while it can be solved in polynomial time for interval graphs~\cite{BB17}.
For claw-free graphs, \prb{ISR} under both $\TJ$ and $\TS$ can be solved in polynomial time~\cite{BonsmaKW14}. 
For bipartite graphs, interestingly, \prb{ISR} is $\NP$-complete under $\TJ$, while $\PSPACE$-complete under $\TS$~\cite{LM19}.
Note that \prb{ISR} and \prb{VCR} under $\TJ$ and $\TS$ are essentially equivalent due to their complementary relationship; therefore, their computational complexities coincide.

In \cite{ISIsoRsuga25}, several results are presented for \prb{ISR} under the $\kTJ$ and $\kTS$ rules.
The paper shows that \prb{ISR} under both $\kTJ$ and $\kTS$ is $\PSPACE$-complete on perfect graphs for any fixed integer $k \geq 2$.
Furthermore, $\kTS$ and $\TS$ are essentially equivalent on even-hole-free graphs~\cite{ISIsoRsuga25}.
As a result, the computational complexity of \prb{ISR} under $\kTS$ on several graph classes contained within the class of even-hole-free graphs follows from the results under $\TS$.

K\v{r}i\v{s}\v{t}an et al.\ studied several reconfiguration problems, including \prb{ISR} and \prb{VCR}, under the \emph{$(k,d)$-Token Jumping} model~\cite{DBLP:GeneralizedTokenJumping}. In this model, $k$ tokens can move simultaneously, and each token can travel within a distance of $d$.
The $(k,d)$-Token Jumping model may appear to generalize the reconfiguration rules $\kTJ$ and $\kTS$; however, the settings are slightly different. The $(k,d)$-Token Jumping model is defined by a bijection between the current configuration and the next configuration. Consequently, a token can move to a vertex currently occupied by another token, as long as the latter token moves to a different vertex in the same step.
In contrast, the definitions of $\kTJ$ and $\kTS$ are based on the symmetric difference between configurations. In these models, a token is prohibited from moving to a vertex that is currently occupied by another token.
Although both the $(k,d)$-Token Jumping model and $\kTJ$ (resp.\ $\kTS$) are natural generalizations of $\TJ$ (resp.\ $\TS$), the difference between them significantly impacts the computational complexity of problems.
In fact, for the number $t$ of tokens, \prb{VCR} under the $(t,1)$-Token Jumping model can be solved in polynomial time~\cite{DBLP:GeneralizedTokenJumping}, while \prb{VCR} under the $\kTS$ is $\PSPACE$-complete when $k=t$~\cite{ISIsoRsuga25}.

\subsection{Our Contribution}\label{subsec:Contribution}
We research the computational complexity of \prb{ISR} and \prb{VCR} under $\kTJ$ and $\kTS$.
An overview of our results is provided in \Cref{ISRresultskTJ,tab:ISRandVCR}.
\revcr{The numbering of theorems follows the system starting from \Cref{sec:guaranvalue}.
Therefore, in this section, note that the numbering does not begin consecutively.}

\paragraph*{The results when parameterized by the guaranteed value}

We first investigate the complexity of \prb{ISR} under $\kTJ$ when $k$ is defined as $|I|-\mu$, where $I$ is the initial independent set of an \prb{ISR} instance and $\mu$ is a parameter referred to as the \emph{guaranteed value}~\cite{DBLP:journals/jal/MahajanR99}.
This parameter measures how far the instance is from the trivial case: when $\mu = 0$ (that is, $k = |I|$), ISR under $\kTJ$ becomes trivial, as all vertices in the independent set can be replaced simultaneously.
Hence, we are interested in the computational complexity of the problem when $\mu$ is small.
It was shown in~\cite{ISIsoRsuga25} that even for any fixed positive integer $\mu$, \textsc{ISR} under $\kTJ$ is $\NP$-complete when $k = |I| - \mu$.
However, the reduction presented in~\cite{ISIsoRsuga25} introduces a large number of edges.
From a practical standpoint, it is particularly important to understand the computational complexity of the problem on sparse graphs, such as planar graphs or graphs with bounded degree.

To answer this question, we present $\NP$-hardness results for these sparse classes.
\setcounter{Introtheorem}{1}
\begin{Introtheorem}\label{Introthm2}
    Let $\mu$ be any fixed positive integer. \prb{ISR} under $\kTJ$ is $\NP$-hard for graphs $G$ of maximum degree $3$ when $k=|I|-\mu\geq 1$, where $I$ is an initial independent set of $G$.
\end{Introtheorem}
\begin{Introtheorem}\label{Introthm3}
    Let $\mu$ be any fixed positive integer. \prb{ISR} under $\kTJ$ is $\NP$-hard for planar graphs $G$ of maximum degree $4$ when $k=|I|-\mu\geq 1$, where $I$ is an initial independent set of $G$.
\end{Introtheorem}
Here, let us explain the main obstacle in the proofs of these theorems.
Consider the case where $\mu = 1$.
If the initial and target independent sets $I$ and $J$ satisfy $|I \cap J| \geq 1$, then the reconfiguration is trivial.
Hence, to ensure that the instance is non-trivial, it is essential to construct $I$ and $J$ so that $I \cap J = \emptyset$.  
A common approach used in existing research is to employ a complete bipartite graph; however, this prevents the resulting graph from being sparse.
A more delicate reduction is required to preserve the sparsity of the graph.
As a key step toward this goal, we introduce a new variant of \prb{Exactly $3$-SAT}, which we call \prb{Internal Exactly $3$-SAT}, and show that it is $\NP$-complete.
In this variant, the input is a $3$-CNF formula that is \rev{promised to be} satisfiable under both the all-true and all-false assignments.
The goal is to determine whether there exists a \emph{mixed} satisfying assignment, that is, a satisfying assignment that is neither all-true nor all-false.
We convert the all-true and all-false assignments to the initial and target independent sets, respectively.
The existence of a mixed satisfying assignment corresponds to the reconfigurability from $I$ to $J$ via an independent set $I'$ such that $|I' \cap I| \geq 1$ and $|I' \cap J| \geq 1$.

We next demonstrate that \prb{ISR} under $\kTJ$ with $k=|I|-\mu$ belongs to $\NP$ even when $\mu = O(\log |I|)$ for some graph classes.
Since many reconfiguration problems are $\PSPACE$-complete due to the potentially super-polynomial length of reconfiguration sequences, showing that a reconfiguration problem belongs to $\NP$ is non-trivial.
It is known that \prb{ISR} under $\kTJ$ with $k=|I|-\mu$ is in $\NP$ when $\mu$ is constant~\cite{ISIsoRsuga25}.
We strengthen this result by proving $\NP$-membership for specific graph classes under the condition $\mu = O(\log |I|)$.

\begin{Introtheorem}\label{Introthm4}
    Let $G$ be an input graph with $n$ vertices, chromatic number $O(1)$ and maximum degree $o(\frac{n}{\log n})$, and let $I$ be an initial independent set of $G$.
    \prb{ISR} under $\kTJ$ is in $\NP$ when $k=|I|-\mu\geq 1$ with any non-negative integer $\mu$ at most $O(\log|I|)$.
\end{Introtheorem}
In the proof of \Cref{Introthm4}, we utilize the concept of an \emph{intersecting family} of a set~\cite{IntersectingTheorem}, which has been extensively studied in the field of \rev{extremal set theory}.
Based on this concept, we derive an upper bound on the length of the reconfiguration sequence.

Here, the following \Cref{Introthm5} constitutes another main result of our work.
\setcounter{Introtheorem}{5}
\begin{Introtheorem}\label{Introthm5}
    \prb{VCR} under $\kTJ$ is in {\XP} for general graphs $G$ when parameterized by $\mu=|S|-k\geq 0$, where $S$ is an initial vertex cover of $G$.
\end{Introtheorem}
Recall that in any graph \( G \), a vertex cover and an independent set are complementary: the complement of a vertex cover of $G$ is an independent set of $G$.  
Thus, \prb{ISR} and \prb{VCR} are generally considered equivalent problems.  
However, our result for \prb{VCR} in \Cref{Introthm5} stands in sharp contrast to the known results for \prb{ISR} (see also \Cref{tab:ISRandVCR}).  
In the proof of \Cref{Introthm5}, we design an \(\XP\) algorithm based on a \emph{clique-compressed reconfiguration graph}~\cite{ISIsoRsuga25}.

\renewcommand{\arraystretch}{1.1}
    \begin{table}[t]
        \caption{The complexity of \prb{ISR} under $\kTJ$ for various graph classes and values of \( k \).  
                Here, \( I \) denotes an initial independent set, and \( \Delta \), \( \mathsf{bw} \), and \( \mathsf{cw} \) denote the maximum degree, bandwidth, and clique-width, respectively, of a given \( n \)-vertex graph. Let \( \varepsilon_0 \in (0,1) \) be a fixed constant.}
        \setlength{\dashlinegap}{10pt}  
        \setlength{\dashlinedash}{1pt} 
        \centering
        \scalebox{0.76}{
        \begin{tabular}{p{2.6cm}|l|l|l|l|l|l}
        \hline
             & \multicolumn{6}{c}{{$\kTJ$}}\\ \hline
            & &   any const.  &  &\multicolumn{2}{l|}{$k=|I|-\mu$} & \\ \cline{5-6}
            &  $k=1~(\TJ)$ &$k\geq 2$ &any $k\leq\varepsilon_0|I|$& $\mu=O(\log |I|)$ & any const.~$\mu>0$  & $k=|I|$  \\ \hline 
            general & \multicolumn{2}{l|}{$\PSPACE$-c.} &$\PSPACE$-c.&open~(in $\NP$?)& $\NP$-c.~\cite{ISIsoRsuga25} & trivially \\ \cline{1-3}\cline{5-6}
            bounded $\Delta$ & $\PSPACE$-c.& $\PSPACE$-c. &[\Cref{ISR_PSPACEcom_largek}]& \multicolumn{1}{p{1.9cm}}{$\NP$-c.} & & always \\ \cline{1-1}
            $\Delta=3$ &\cite{HD05,DBLP:conf/iwpec/Zanden15,W18} & [\Cref{PSPACE-comp-planar}]  & & \multicolumn{1}{p{1.9cm}}{[\Cref{ISR_NPcom_maxdeg3}]} & & yes \\ \cline{1-1}\cdashline{2-3} \cline{4-4}
            planar $\cap$ \newline $\Delta=o(\frac{n}{\log n})$ & & & $\NP$-h. & \multicolumn{1}{p{1.9cm}}{[\Cref{ISR_NPcom_planarmaxdeg4}] \newline [\Cref{ISRinNPlogmu}]} &  &  \\ \cline{1-1}
            planar $\cap$ $\Delta=4$ & &&&\multicolumn{1}{p{1.9cm}}{}&  & \\ \cline{1-1}\cdashline{2-3} \cline{4-6}
            planar $\cap$ $\Delta=3$ & &&open&\multicolumn{2}{l|}{open~($\NP$-hard?)} &  \\ \cline{1-1}\cdashline{2-4} \cline{5-6}
            planar $\cap$ $\Delta=3$ & & && \multicolumn{2}{l|}{$\XP$} &  \\ 
            $\cap$ bounded $\mathsf{bw}$ &  & &&\multicolumn{2}{l|}{parameterized} &  \\ \cline{1-1}
            bounded $\mathsf{cw}$ & &&&\multicolumn{2}{l|}{by $\mu=|I|-k$} &  \\ \cline{1-3} 
            perfect & $\PSPACE$-c.~\cite{KMM12} & $\PSPACE$-c.~\cite{ISIsoRsuga25} &&\multicolumn{2}{l|}{\cite{ISIsoRsuga25}} &  \\ \cline{1-3}
            bipartite & $\NP$-c.~\cite{LM19}  & open &&\multicolumn{2}{l|}{} &  \\ \cline{1-3}
            claw-free & {\P}~\cite{BonsmaKW14,KMM12} & $\PSPACE$-c. ($k=2$) &&\multicolumn{2}{l|}{} &  \\ \cline{1-1}
            line &  & [\Cref{PSPACE-comp-claw-free}] && \multicolumn{2}{l|}{} &  \\ \hline
            even-hole-free & \multicolumn{6}{l}{always yes~\cite{KMM12}} \\ \hline
        \end{tabular}
        }
        \label{ISRresultskTJ}
    \end{table}

    \begin{table}[t]
        \caption{Comparison between the complexity of \prb{ISR} and \prb{VCR} under $\kTJ$ on general graphs and graphs of maximum degree $3$. The vertex subset \( A \) of the input graph represents the initial solution for both \prb{ISR} and \prb{VCR}. Specifically, \( A \) is the input independent set for \prb{ISR}, and the input vertex cover for \prb{VCR}. Let $\varepsilon_0\in (0,1)$ be some constant.}
        \centering
        \scalebox{0.9}{
        \begin{tabular}{c|c|c|c|c}
        \hline
              &  & \multirow{2}{*}{any $k\leq \varepsilon_0|A|$} & \multicolumn{2}{c}{$k=|A|-\mu$} \\ \cline{4-5}
             problems& graph classes& & $\mu=O(\log |A|)$ & any fixed $\mu\geq1$ \\ \hline
            \multirow{2}{*}{ISR} & general & $\PSPACE$-c. & open~(in $\NP$?) & $\NP$-c.~\cite{ISIsoRsuga25} \\ \cline{2-2}\cline{4-5}
             & maximum degree $3$ & [\Cref{ISR_PSPACEcom_largek}] & \multicolumn{2}{c}{$\NP$-c.~[\Cref{ISR_NPcom_maxdeg3},~\Cref{ISRinNPlogmu}]} \\ \hline\hline
             \multirow{2}{*}{VCR}& general & $\PSPACE$-c. & \multicolumn{2}{c}{$\XP$~[\Cref{VCR_XP}]} \\\cline{2-2}
             & maximum degree $3$ &[\Cref{VCR_PSPACEcom_largek}] & \multicolumn{2}{c}{parameterized by $\mu=|A|-k$} \\ \hline
        \end{tabular}
        }
        \label{tab:ISRandVCR}
    \end{table}

    \paragraph*{The results when $k$ is constant}
We establish the \(\PSPACE\)-completeness of \prb{ISR} for fixed \(k\) on several graph classes.
\begin{Introtheorem}\label{Introthm7}
    Let $k\geq 2$ be any fixed positive integer. \prb{ISR} under $\Rule\in\{\kTS,\kTJ\}$ is $\PSPACE$-complete for planar graphs of maximum degree $3$ and bounded bandwidth.
\end{Introtheorem}
Together with known results for the case \(k = 1\)~\cite{HD05,DBLP:conf/iwpec/Zanden15,W18}, our findings provide a complete characterization of the \(\PSPACE\)-completeness of \prb{ISR} under \(\kTJ\) and \(\kTS\) for every fixed positive integers~\(k\) and planar graphs of maximum degree $3$ and bounded bandwidth.

We further demonstrate the following \Cref{Introthm8}.
\begin{Introtheorem}\label{Introthm8}
    \prb{ISR} under $\Rule\in \{2\text{-}\TJ, 2\text{-}\TS\}$ is $\PSPACE$-complete for line graphs.
\end{Introtheorem}
This result stands in sharp contrast to the polynomial-time solvability of \prb{ISR} under \(\TJ\) and \(\TS\) on claw-free graphs~\cite{BonsmaKW14}, which include line graphs as a special case.

\paragraph*{The results when $k$ is superconstant}
We investigate the complexity of \prb{ISR} under $\kTJ$ when $k$ is superconstant in the initial independent set size $|I|$.
To this end, we construct a polynomial-time reduction from the ``optimization variant'' called \prb{MaxminISR} to \prb{ISR}.
By using the $\PSPACE$-hardness of approximating \prb{MaxminISR}~\cite{DBLP:conf/stoc/HiraharaO24,karthikc.s.2023inapproximability,DBLP:conf/stacs/Ohsaka23}, we prove that \prb{ISR} under \(\kTJ\) is \(\PSPACE\)-complete on graphs of maximum degree~3 for a wide range of values of \(k\), including \(k = O(1)\), \(\Theta(\log |I|)\), \(\Theta(|I|^{O(1)})\), and \(\Theta(|I|)\).

In particular, we show the following \Cref{Introthm9}.

\setcounter{Introtheorem}{8}
\begin{Introtheorem}\label{Introthm9}
    There exists some constant $\varepsilon_0\in(0,1)$ such that \prb{ISR} under $\kTJ$ on graphs of maximum degree $3$ is $\PSPACE$-complete for any $k$ satisfying the following condition:
    there exists a constant $c$ such that $k\leq \varepsilon_0|I|$ holds whenever $|I|\geq c$, where $I$ is the initial independent set of the input graph. 
\end{Introtheorem}

We show that a similar result holds for \prb{VCR}.
\setcounter{Introtheorem}{11}
\begin{Introtheorem}\label{Introthm12}
    There exists some constant $\varepsilon_0\in(0,1)$ such that \prb{VCR} under $\kTJ$ on graphs of maximum degree $3$ is $\PSPACE$-complete for any $k$ satisfying the following condition:
    there exists a constant $c$ such that $k\leq \varepsilon_0|S|$ holds whenever $|S|\geq c$, where $S$ is the initial vertex cover of the input graph.
\end{Introtheorem}

\paragraph{Outline.}
The remainder of this paper is organized as follows.  
We begin with preliminaries in \Cref{sec:Preliminaries}.  
In \Cref{sec:guaranvalue}, we study the problems \prb{ISR} and \prb{VCR} under $\kTJ$, focusing on the guaranteed value $\mu$.  
In \Cref{sec:k_const}, we demonstrate the \(\PSPACE\)-completeness of these problems when \(k\) is fixed.  
Finally, in \Cref{sec:large_k}, we analyze the case where \(k\) is superconstant.

\section{Preliminaries}\label{sec:Preliminaries}
For a positive integer $i$, we write $[i] = \{1,2,\ldots,i\}$. 
Let $G=(V,E)$ be a finite, simple, and undirected graph with the vertex set $V$ and the edge set $E$. 
We use $V(G)$ and $E(G)$ to denote the vertex set and the edge set of $G$, respectively. 
For a vertex $v$ of $G$, $N_G(v)$ and $N_G[v]$ denote the \emph{open neighborhood} and the \emph{closed neighborhood} of $v$, respectively, that is, $N_G(v)=\{u\in V\colon uv\in E\}$ and $N_G[v]=N_G(v)\cup \{v\}$. 
For a vertex set $X\subseteq V(G)$, we define $N_G(X)=\{v\notin X\colon uv\in E(G), u\in X\}$ and $N_G[X]=N_G(X)\cup X$.
A subgraph of $G$ is a graph $G'$ such that $V(G')\subseteq V(G)$ and $E(G') \subseteq E(G)$.
For a subset $S\subseteq V(G)$, $G[S]$ denotes the subgraph induced by $S$.
The \emph{degree} of a vertex $v$ in a graph $G$ is the number of edges incident to $v$.  
The \emph{maximum degree} of $G$ is the largest degree among all vertices in $G$.
A sequence $\langle v_0, v_1, \ldots, v_\ell \rangle_{G}$ of vertices of a graph $G$, where $v_{i-1}v_i \in E$ for each $i \in [\ell]$, is called a \emph{path} if all the vertices are distinct. It is called a \emph{cycle} if $v_0, v_1, \ldots, v_{\ell-1}$ are distinct and $v_0 = v_\ell$. The value $\ell$ is referred to as the \emph{length} of the path or cycle.
A graph $G$ is said to be \emph{connected} if there exists a path between $u$ and $v$ for any pair of vertices $u,v \in V(G)$. 
The \emph{chromatic number} of $G$ is the smallest positive integer $c$ such that $G$ has a \emph{$c$-coloring}, where a $c$-coloring of a graph $G$ is a mapping $f\colon V(G)\rightarrow[c]$ such that $f(u)\neq f(v)$ if $uv\in E(G)$.
The \emph{bandwidth} of a graph \( G = (V, E) \) is the minimum integer \( b \) such that there exists a bijection \( f\colon V \to \{1, \dots, |V|\} \) satisfying \( |f(u) - f(v)| \leq b \) for every edge \( uv \in E \).

We conclude this section with the following simple remark. 
We can observe that the problems \prb{ISR} and \prb{VCR} can be solved in nondeterministic polynomial space. 
Thus, by applying Savitch's Theorem~\cite{S70}, these problems are in $\PSPACE$. 
Consequently, to establish $\PSPACE$-completeness of \prb{ISR} and \prb{VCR}, it suffices to prove $\PSPACE$-hardness.

\section{When Parameterized by a Guaranteed Value}\label{sec:guaranvalue}
In this section, we discuss \prb{ISR} and \prb{VCR} under $\kTJ$, focusing on the guaranteed value.

\subsection{ISR}\label{ISRguaranteed}
We \rev{investigate} the $\NP$-completeness of \prb{ISR} under $\kTJ$ when $k=|I|-\mu$, \rev{where $I$ is the initial independent set of an \prb{ISR} instance, and $\mu$ is a parameter called the \emph{guaranteed value}.
Hereafter, we use $I$ to denote the initial independent set of an \prb{ISR} instance.}
For any fixed positive integer $\mu$, it is known that \prb{ISR} under $\kTJ$ is $\NP$-complete when $k = |I| - \mu$~\cite{ISIsoRsuga25}.  
In \Cref{subsubsec:ISR_NPhard}, we show that \prb{ISR} under $\kTJ$ remains $\NP$-hard even when the input graph is restricted to graphs of maximum degree 3, or to planar graphs of maximum degree 4, where $k = |I| - \mu$ for any fixed positive integer $\mu$.
Then, in \Cref{subsubsec:ISRmemberNP}, we prove that the problem remains in $\NP$ even when $\mu = O(\log |I|)$ and an input graph is a graph of bounded chromatic number and maximum degree $o(\frac{n}{\log n})$, \rev{where $n$ is the number of vertices in the input graph}.

\subsubsection{NP-hardness}\label{subsubsec:ISR_NPhard}
\rev{We show the $\NP$-hardness of \prb{ISR} under $\kTJ$ for graphs of maximum degree $3$ and planar graphs of maximum degree $4$ when $k = |I| - \mu$ with any fixed positive integer $\mu$.}
\rev{To this end, we construct a chain of reductions.}
As a source problem of our reduction, we will use \prb{Exactly 3-SAT} (\prb{E3-SAT} for short), which is $\NP$-complete~\cite{GJ79}.
Firstly, \prb{E3-SAT} is reduced to \prb{Internal Exactly 3-SAT} (\prb{IntE3-SAT} for short), a new variant of \prb{E3-SAT} (to the best of our knowledge).
Afterward, \prb{IntE3-SAT} is reduced to our problem.

Here, let us define \prb{E3-SAT} and \prb{IntE3-SAT}. 
Let $\phi=C_1\land C_2\land \cdots \land C_m$ be a Boolean formula in conjunctive normal form (CNF), and $X=\{x_1,\ldots,x_n\}$ be \rev{the} variable set of $\phi$. 
Each clause $C_i$ \rev{for $i\in[m]$ in $\phi$} is a disjunction of literals and each literal appears as either a positive form $x_j$ or a negative form $\overline{x_j}$ for a variable $x_j\in X$. 
A variable assignment is a mapping $b\colon X \to \{\True, \False\}$.
We \rev{say} that a variable assignment \emph{satisfies} $\phi$ if $\phi$ evaluates to $\True$.
A CNF formula $\phi$ is called an E$k$-CNF formula if any clause in $\phi$ consists of exactly $k$ literals. Given an E3-CNF \rev{formula} $\phi$, \prb{E3-SAT} asks whether there exists a variable assignment that satisfies $\phi$.

A variable assignment $b$ is called an \emph{all-$\True$ assignment} if $b(x) = \True$ for all $x\in X$, and an \emph{all-$\False$ assignment} if $b(x) = \False$ for all $x\in X$.
Otherwise, it is called \emph{mixed}.
An E$k$-CNF formula $\phi$ is called \emph{sandwiched} if all-$\True$ and all-$\False$ assignments satisfy $\phi$.
In other words, any clause of a sandwiched E$k$-CNF formula $\phi$ contains both positive and negative literals.
Given a sandwiched \rev{E$k$-CNF} formula $\phi$, \rev{\prb{IntE$k$-SAT}} asks whether there exists a mixed variable assignment that satisfies $\phi$.
\rev{We first show that \prb{IntE3-SAT} is $\NP$-complete by reducing \prb{E3-SAT}.}

\begin{lemma}\label{IntSAT_NPcomp}
    \prb{IntE3-SAT} is $\NP$-complete.
\end{lemma}

\revcr{For the proof, refer to Section~\ref{appendixsubsec:IntSAT_NPcomp}.}
By the polynomial-time reduction from \prb{IntE3-SAT}, we prove the following \Cref{ISR_NPcom_maxdeg3}.

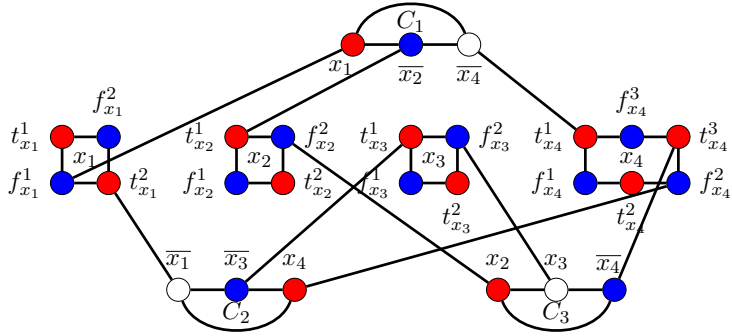
\begin{figure}[t]
    \centering
    \scalebox{0.85}{
\begin{tikzpicture}[scale=0.9]
\node[] at (0.4,0) {$x_1$};
\node[fill=red!100, draw=black, circle, minimum size=2.5mm, label=left:$t^1_{x_1}$] (x11) at (0,0.4){};
\node[fill=blue!100, draw=black, circle, minimum size=2.5mm, label=left:$f^1_{x_1}$] (x12) at (0,-0.4){};
\node[fill=red!100, draw=black, circle, minimum size=2.5mm, label=right:$t^2_{x_1}$] (x13) at (0.8,-0.4){};
\node[fill=blue!100, draw=black, circle, minimum size=2.5mm, label=above:$f^2_{x_1}$] (x14) at (0.8,0.4){};
\draw[very thick] (x11)--(x12);
\draw[very thick] (x12)--(x13);
\draw[very thick] (x14)--(x13);
\draw[very thick] (x11)--(x14);

\node[] at (3+0.4,0) {$x_2$};
\node[fill=red!100, draw=black, circle, minimum size=2.5mm, label=left:$t^1_{x_2}$] (x21) at (3+0,0.4){};
\node[fill=blue!100, draw=black, circle, minimum size=2.5mm, label=left:$f^1_{x_2}$] (x22) at (3+0,-0.4){};
\node[fill=red!100, draw=black, circle, minimum size=2.5mm, label=right:$t^2_{x_2}$] (x23) at (3+0.8,-0.4){};
\node[fill=blue!100, draw=black, circle, minimum size=2.5mm, label=right:$f^2_{x_2}$] (x24) at (3+0.8,0.4){};
\draw[very thick] (x21)--(x22);
\draw[very thick] (x22)--(x23);
\draw[very thick] (x24)--(x23);
\draw[very thick] (x21)--(x24);

\node[] at (6+0.4,0) {$x_3$};
\node[fill=red!100, draw=black, circle, minimum size=2.5mm, label=left:$t^1_{x_3}$] (x31) at (6+0,0.4){};
\node[fill=blue!100, draw=black, circle, minimum size=2.5mm, label=left:$f^1_{x_3}$] (x32) at (6+0,-0.4){};
\node[fill=red!100, draw=black, circle, minimum size=2.5mm, label=below:$t^2_{x_3}$] (x33) at (6+0.8,-0.4){};
\node[fill=blue!100, draw=black, circle, minimum size=2.5mm, label=right:$f^2_{x_3}$] (x34) at (6+0.8,0.4){};
\draw[very thick] (x31)--(x32);
\draw[very thick] (x32)--(x33);
\draw[very thick] (x34)--(x33);
\draw[very thick] (x31)--(x34);

\node[] at (9.8,0) {$x_4$};
\node[fill=red!100, draw=black, circle, minimum size=2.5mm, label=left:$t^1_{x_4}$] (x41) at (9+0,0.4){};
\node[fill=blue!100, draw=black, circle, minimum size=2.5mm, label=left:$f^1_{x_4}$] (x42) at (9+0,-0.4){};
\node[fill=red!100, draw=black, circle, minimum size=2.5mm, label=below:$t^2_{x_4}$] (x43) at (9+0.8,-0.4){};
\node[fill=blue!100, draw=black, circle, minimum size=2.5mm, label=right:$f^2_{x_4}$] (x44) at (9+1.6,-0.4){};
\node[fill=red!100, draw=black, circle, minimum size=2.5mm, label=right:$t^3_{x_4}$] (x45) at (9+1.6,0.4){};
\node[fill=blue!100, draw=black, circle, minimum size=2.5mm, label=above:$f^3_{x_4}$] (x46) at (9+0.8,0.4){};
\draw[very thick] (x41)--(x42);
\draw[very thick] (x42)--(x43);
\draw[very thick] (x44)--(x43);
\draw[very thick] (x45)--(x44);
\draw[very thick] (x45)--(x46);
\draw[very thick] (x46)--(x41);

\node[] at (6,2.4) {$C_1$};
\node[fill=red!100, draw=black, circle, minimum size=2.5mm] (c11) at (5,2){};
\node[] at (4.8,1.6) {$x_1$};
\node[fill=blue!100, draw=black, circle, minimum size=2.5mm, label=below:$\overline{x_2}$] (c12) at (6,2){};
\node[fill=white!10, draw=black, circle, minimum size=2.5mm, label=below:$\overline{x_4}$] (c13) at (7,2){};
\draw[very thick] (c11)--(c12)--(c13) to[out=110, in=70] (c11);

\node[] at (3,-2.65) {$C_2$};
\node[fill=white!100, draw=black, circle, minimum size=2.5mm, label=above:$\overline{x_1}$] (c21) at (2,-2.25){};
\node[fill=blue!100, draw=black, circle, minimum size=2.5mm, label=above:$\overline{x_3}$] (c22) at (3,-2.25){};
\node[fill=red!100, draw=black, circle, minimum size=2.5mm, label=above:$x_4$] (c23) at (4,-2.25){};
\draw[very thick] (c21)--(c22)--(c23) to[out=-110, in=-70] (c21);

\node[] at (8.5,-2.65) {$C_3$};
\node[fill=red!100, draw=black, circle, minimum size=2.5mm, label=above:$x_2$] (c31) at (7.5,-2.25){};
\node[fill=white!100, draw=black, circle, minimum size=2.5mm, label=above:$x_3$] (c32) at (8.5,-2.25){};
\node[fill=blue!100, draw=black, circle, minimum size=2.5mm] (c33) at (9.5,-2.25){};
\node[] at (9.4, -1.8) {$\overline{x_4}$};
\draw[very thick] (c31)--(c32)--(c33) to[out=-110, in=-70] (c31);

\draw[very thick] (x12)--(c11);
\draw[very thick] (x21)--(c12);
\draw[very thick] (x41)--(c13);

\draw[very thick] (x13)--(c21);
\draw[very thick] (x31)--(c22);
\draw[very thick] (x44)--(c23);

\draw[very thick] (x24)--(c31);
\draw[very thick] (x34)--(c32);
\draw[very thick] (x45)--(c33);

\end{tikzpicture}}
    \caption{Illustration of the construction of $G$ from a boolean formula $\phi'$ used in the proof of \Cref{ISR_NPcom_maxdeg3}, where $\phi'=(x_1\lor \overline{x_2}\lor\overline{x_4})\land(\overline{x_1}\lor \overline{x_3}\lor x_4)\land (x_2\lor x_3\lor \overline{x_4})$. The blue marked tokens are on $I$, and the red marked tokens are on $J$.}
    \label{fig:ISRmaxdg3}
\end{figure}

\begin{theorem}\label{ISR_NPcom_maxdeg3}
    Let $\mu$ be any fixed positive integer. \prb{ISR} under $\kTJ$ is $\NP$-hard for graphs $G$ of maximum degree $3$ when $k=|I|-\mu\geq 1$, where $I$ is an initial independent set of $G$.
\end{theorem}
\begin{proof}
    We use a polynomial-time reduction from \prb{IntE3-SAT}. 
    Let $\phi'$ be an instance of \prb{IntE3-SAT} and $X'$ be \rev{the} set of variables of $\phi'$.
    We will construct an instance $(G,I,J,\kTJ)$ of \prb{ISR} under $\kTJ$ where $k=|I|-1$ (see \Cref{fig:ISRmaxdg3} for an illustration), and then modify for any fixed $\mu\geq1$.
    
    For each variable $x\in X'$, let $a_x$ denote the number of clauses of $\phi'$ in which $x$ appears as a literal.
    For each variable $x\in X'$, we set up a \emph{variable gadget} $Y_x$, which is defined as a cycle with $2a_x$ vertices~(note that if $x$ appears only once, then $Y_x$ is a path with two vertices).
    The vertices in $Y_x$ are labeled with $t_x^1, f_x^1,t_x^2, f_x^2,\ldots, t_x^{a_x}, f_x^{a_x}$ in a counterclockwise order. 
    \rev{Intuitively,} $t_x^i$ and $f_x^i$ for each $i\in [a_x]$ correspond to $\True$ and $\False$ assignments for $x$, respectively.
    We refer to a vertex \rev{of $Y_x$} corresponding to $\True$ (resp.\ $\False$) assignment to $x$ as a \emph{true vertex} (resp.\ \emph{false vertex}).
    For each clause $C$, \rev{we set up a \emph{clause gadget} $K_C$, which is defined as a complete graph with three vertices.}
    The vertices in $K_C$ correspond to \rev{the three} literals in $C$. 
    We refer to a vertex \rev{of $K_C$} corresponding to a positive literal (resp.\ negative literal) \rev{of $C$} as a \emph{positive vertex} (resp.\ \emph{negative vertex}).
    Finally, we connect variable gadgets and clause gadgets with edges as follows:
    For a vertex $v$ in a clause gadget corresponding to a literal $l$ of a variable $x$ with occurrence $i \in [a_x]$, connect $v$ to $t_x^{i}$ if the literal is negative; otherwise, connect $v$ to $f_x^{i}$.
    Let $G$ be the obtained graph.
    Observe that $G$ is a graph of maximum degree~$3$.
    Let $I$ be a set of all false vertices in variable gadgets and a \rev{negative} vertex arbitrarily chosen from each clause gadget. 
    \rev{Let} $J$ be a set of all true vertices in variable gadgets and a \rev{positive} vertex arbitrarily chosen from each clause gadget. 
    Note that such $I$ and $J$ always exist as each clause has both positive and negative literals due to the definition of \prb{IntE3-SAT}. 
    Finally, we set $k=|I|-1$.
    
    If $\mu \geq 2$, then we also add $\mu-1$ isolated vertices to $G$.
    Let $V^*$ be the set of all added vertices.
    Then let $G^*$ be an obtained graph, $I^* = I \cup V^*$, and $J^* = J \cup V^*$.
    We also set $k=|I^*|-\mu$.
    Note that $I$ is a maximum independent set of $G$ and thus any independent set with size $|I^*|=|I|+|V^*|$ of $G^*$ contains $V^*$.
    This allows us to move at most $k=|I^*|-\mu =  |I|+|V^*| -\mu = |I|-1$ tokens simultaneously only on $G$.
    Therefore, $(G,I,J,(|I|-1)\text{-}\TJ)$ is a yes-instance if and only if $(G^*,I^*,J^*,(|I^*|-\mu)\text{-}\TJ)$ is a yes-instance.
    
    For this reason, we discuss only the case when $\mu=1$ here.
    The instance $(G,I,J,\kTJ)$ of \prb{ISR} under $\kTJ$ is clearly obtained in polynomial time.
    To complete our reduction, we will show that $\phi'$ is a yes-instance of \prb{IntE3-SAT} if and only if $(G, I, J,\kTJ)$ is a yes-instance of \prb{ISR} under $\kTJ$ where $k=|I|-1$. 

    Assume that there is a variable assignment $b'$ that is mixed and satisfies $\phi'$. 
    We construct an independent set $I'$ of $G$ as follows.
    For a variable $x_i$ with $i\in[n]$, if $x_i$ is assigned $\True$, then all true vertices of $Y_{x_i}$ are contained in $I'$.
    Similarly, if $x_i$ is assigned $\False$, then all false vertices of $Y_{x_i}$ are contained in $I'$.
    Since those vertices are chosen according to $b'$, one vertex in each clause gadget that corresponds to a literal evaluating $\True$ can also be contained in $I'$.
    Furthermore, since $b'$ is mixed,
    we have $|I\cap I'|\geq 1$ and $|I'\cap J|\geq 1$.
    Therefore, a reconfiguration sequence $\sigma=\langle I, I', J \rangle$ exists.

    Conversely, assume that there is a reconfiguration sequence $\sigma=\langle I=I_0, I_1,\ldots, I_\ell=J\rangle$. 
    Consider an integer $i\in [\ell]$ such that tokens on $I_i$ are placed on true vertices for the first time.
    In other words, all tokens of $I_{i-1}$ in variable gadgets are on false vertices.
    Since any positive vertex of any clause gadget is adjacent to a false vertex of a variable gadget, all positive vertices are not in $I_{i-1}$.
    Then, we claim that not all true vertices are in $I_i$.
    For the sake of contradiction, assume that all true vertices are in $I_i$.
    Then, since each negative vertex is adjacent to a true vertex, each token of $I_i$ in clause gadgets is on a positive vertex. 
    Thus, since $I_{i-1}$ contains only false vertices and negative vertices, $I_{i-1}\cap I_i=\emptyset$. 
    However, $|I_{i-1}\cap I_{i}|\geq 1$ must hold because they are adjacent under $\kTJ$, which is a contradiction.
    Therefore, we conclude that $I_i$ contains both true and false vertices.
    From our construction of the variable gadgets, either true or false vertices are in $I_i$ for each variable gadget.
    For each $j\in[n]$, let $b'$ be a variable assignment such that $x_j$ is assigned $\True$ if and only if true vertices of $Y_{x_j}$ are in $I_i$. 
    Then, literals of clauses in $\phi'$ corresponding to vertices in $I_i$ evaluate $\True$ from our construction, that is, $b'$ satisfies $\phi'$.
    Therefore, $\phi'$ is a yes-instance of \prb{IntE3-SAT}.
\end{proof}

By \Cref{ISR_NPcom_maxdeg3}, we have proven the $\NP$-hardness of \prb{ISR} under $\kTJ$ on graphs of maximum degree 3 unless we can move all tokens. 
We now proceed to the $\NP$-hardness of \prb{ISR} under $\kTJ$ on planar graphs of maximum degree $4$.
The constructed graph $G$ in the proof of \Cref{ISR_NPcom_maxdeg3} may have some edges crossing on a plane.
We will eliminate these crossings by replacing them with a \emph{crossover gadget}~(as shown in \Cref{fig:CrossGadget}).
A crossover gadget consists of eight vertices $u'_1, u'_2, v'_1, v'_2, w_1, \ldots, w_4$ and twelve edges: $u'_1w_1$, $u'_1w_4$, $u'_2w_2$, $u'_2w_3$, $v'_1w_1$, $v'_1w_2$, $v'_2w_3$, $v'_2w_4$, $w_1w_2$, $w_2w_3$, $w_3w_4$, and $w_4w_1$.
For two crossing edges $u_1u_2$ and $v_1v_2$ of $G$, \emph{replacing} this crossing with a crossover gadget means removing the edges $u_1u_2$ and $v_1v_2$, inserting a crossover gadget, and adding the edges $u_1u'_1$, $u_2u'_2$, $v_1v'_1$, and $v_2v'_2$.
The crossover gadget always contains exactly three tokens corresponding to a maximum independent set of the gadget.

\begin{figure}[t]
    \centering
    \scalebox{0.75}{\begin{tikzpicture}[scale=1.1]

\node[fill=white!10, draw=black, circle, minimum size=4.5mm, label=90:\normalsize {$u_1$}] (Au1) at (-4,0){};
\node[fill=white!10, draw=black, circle, minimum size=4.5mm, label=270:\normalsize {$u_2$}] (Au2) at (-2,0){};
\node[fill=white!10, draw=black, circle, minimum size=4.5mm, label=180:\normalsize {$v_1$}] (Av1) at (-3,1){};
\node[fill=white!10, draw=black, circle, minimum size=4.5mm, label=0:\normalsize {$v_2$}] (Av2) at (-3,-1){};
\draw[very thick] (Au1)--(Au2);
\draw[very thick] (Av1)--(Av2);

\draw[very thick, dashed] (Au1)--(-4.75,0);
\draw[very thick, dashed] (Au2)--(-1.25,0);
\draw[very thick, dashed] (Av1)--(-3,1.75);
\draw[very thick, dashed] (Av2)--(-3,-1.75);

\node[fill=black!100, draw=black, circle, minimum size=2.5mm] (Au1token) at (-4,0){};
\node[fill=black!100, draw=black, circle, minimum size=2.5mm] (Av1token) at (-3,1){};

\draw (-3,0) [dotted] circle [radius=0.7];

\node[fill=white!10, draw=black, circle, minimum size=4.5mm, label=90:\normalsize {$u_1$}] (Out1) at (2,0){};
\node[fill=white!10, draw=black, circle, minimum size=4.5mm, label=270:\normalsize {$u_2$}] (Out2) at (6.5,0){};
\node[fill=white!10, draw=black, circle, minimum size=4.5mm, label=180:\normalsize {$v_1$}] (Out3) at (4.25,2.25){};
\node[fill=white!10, draw=black, circle, minimum size=4.5mm, label=0:\normalsize {$v_2$}] (Out4) at (4.25,-2.25){};

\node[fill=white!10, draw=black, circle, minimum size=4.5mm, label=90:\normalsize {$u'_1$}] (Mid1) at (2.75,0){};
\node[fill=white!10, draw=black, circle, minimum size=4.5mm, label=270:\normalsize {$u'_2$}] (Mid2) at (5.75,0){};
\node[fill=white!10, draw=black, circle, minimum size=4.5mm, label=180:\normalsize {$v'_1$}] (Mid3) at (4.25,1.5){};
\node[fill=white!10, draw=black, circle, minimum size=4.5mm, label=0:\normalsize {$v'_2$}] (Mid4) at (4.25,-1.5){};

\node[fill=white!10, draw=black, circle, minimum size=4.5mm, label=135:\normalsize {$w_1$}] (In1) at (3.75,0.5){};
\node[fill=white!10, draw=black, circle, minimum size=4.5mm, label=225:\normalsize {$w_4$}] (In2) at (3.75,-0.5){};
\node[fill=white!10, draw=black, circle, minimum size=4.5mm, label=45:\normalsize {$w_2$}] (In3) at (4.75,0.5){};
\node[fill=white!10, draw=black, circle, minimum size=4.5mm, label=315:\normalsize {$w_3$}] (In4) at (4.75,-0.5){};

\draw[very thick, dashed] (Out1)--(1.25,0);
\draw[very thick, dashed] (Out2)--(7.25,0);
\draw[very thick, dashed] (Out3)--(4.25,3);
\draw[very thick, dashed] (Out4)--(4.25,-3);

\draw[very thick] (Out1)--(Mid1);
\draw[very thick] (Out2)--(Mid2);
\draw[very thick] (Out3)--(Mid3);
\draw[very thick] (Out4)--(Mid4);

\draw[very thick] (Mid1)--(In1);
\draw[very thick] (Mid1)--(In2);
\draw[very thick] (Mid2)--(In3);
\draw[very thick] (Mid2)--(In4);
\draw[very thick] (Mid3)--(In1);
\draw[very thick] (Mid3)--(In3);
\draw[very thick] (Mid4)--(In2);
\draw[very thick] (Mid4)--(In4);

\draw[very thick] (In1)--(In3);
\draw[very thick] (In2)--(In4);
\draw[very thick] (In4)--(In3);
\draw[very thick] (In1)--(In2);

\node[fill=black!100, draw=black, circle, minimum size=2.5mm] (Out1token) at (2,0){};
\node[fill=black!100, draw=black, circle, minimum size=2.5mm] (Out3token) at (4.25,2.25){};
\node[fill=black!100, draw=black, circle, minimum size=2.5mm] (In1token) at (3.75,0.5){};
\node[fill=black!100, draw=black, circle, minimum size=2.5mm] (Mid2token) at (5.75,0){};
\node[fill=black!100, draw=black, circle, minimum size=2.5mm] (Mid4token) at (4.25,-1.5){};
\draw [->, >={Triangle[width=10mm, length=5mm]}, line width=6mm, gray] (-0.5,0) -- (0.5,0);

\draw (4.25,0) [dotted] circle [radius=1.9];






\end{tikzpicture}}
    \caption{An illustration of replacing the crossing edges $u_1u_2$ and $v_1v_2$ with a crossover gadget. The gadget contains eight vertices and three tokens, and its maximum degree is $4$.}
    \label{fig:CrossGadget}
\end{figure}
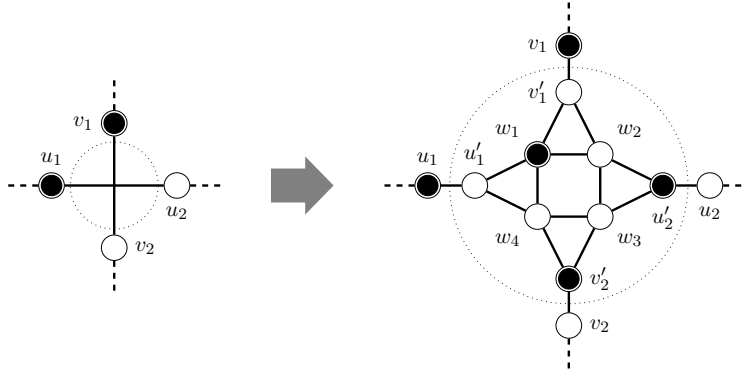
\begin{figure}[t]
    \centering
    \begin{tikzpicture}[scale=1]

\draw[dotted] (0.5,0.5)--(10,0.5);
\draw[dotted] (0.5,1)--(10,1);
\draw[dotted] (0.5,1.5)--(10,1.5);
\draw[dotted] (0.5,2)--(10,2);
\draw[dotted] (0.5,2.5)--(10,2.5);
\draw[dotted] (0.5,3)--(10,3);
\draw[dotted] (0.5,3.5)--(10,3.5);
\draw[dotted] (0.5,4)--(10,4);
\draw[dotted] (0.5,4.5)--(10,4.5);
\draw[dotted] (0.5,5)--(10,5);
\draw[dotted] (0.5,5.5)--(10,5.5);

\foreach \i in {1,...,11} {
    \node[] at (0.25,\i*0.5) {$\i$};
}


\draw[dotted] (1,0)--(1,6);
\draw[dotted] (1.5,0)--(1.5,6);
\draw[dotted] (2,0)--(2,6);
\draw[dotted] (2.5,0)--(2.5,6);
\draw[dotted] (3,0)--(3,6);
\draw[dotted] (3.5,0)--(3.5,6);
\draw[dotted] (4,0)--(4,6);
\draw[dotted] (4.5,0)--(4.5,6);
\draw[dotted] (5,0)--(5,6);
\draw[dotted] (5.5,0)--(5.5,6);
\draw[dotted] (6,0)--(6,6);
\draw[dotted] (6.5,0)--(6.5,6);
\draw[dotted] (7,0)--(7,6);
\draw[dotted] (7.5,0)--(7.5,6);
\draw[dotted] (8,0)--(8,6);
\draw[dotted] (8.5,0)--(8.5,6);
\draw[dotted] (9,0)--(9,6);
\draw[dotted] (9.5,0)--(9.5,6);

\foreach \i in {1,...,18} {
  \node at (0.5*\i+0.5, -0.25) {$\i$};
}


\node[] at (1.5,0.75) {$C_1$};
\node[fill=white!10, draw=black, circle, minimum size=2.5mm] (C11) at (1,1){};
\node[fill=white!10, draw=black, circle, minimum size=2.5mm] (C12) at (2,1){};
\node[fill=white!10, draw=black, circle, minimum size=2.5mm] (C13) at (3,1){};
\draw[very thick] (C11)--(C12);
\draw[very thick] (C12)--(C13);
\draw[very thick] (C11)--(1,0.5)--(3,0.5)--(C13);

\node[fill=red!100, draw=black, circle, minimum size=2.5mm] at (1,1){};
\node[fill=blue!100, draw=black, circle, minimum size=2.5mm] at (2,1){};

\node[] at (4.5,0.75) {$C_2$};
\node[fill=white!10, draw=black, circle, minimum size=2.5mm] (C21) at (4,1){};
\node[fill=white!10, draw=black, circle, minimum size=2.5mm] (C22) at (5,1){};
\node[fill=white!10, draw=black, circle, minimum size=2.5mm] (C23) at (6,1){};
\draw[very thick] (C21)--(C22);
\draw[very thick] (C22)--(C23);
\draw[very thick] (C21)--(4,0.5)--(6,0.5)--(C23);

\node[fill=red!100, draw=black, circle, minimum size=2.5mm] at (6,1){};
\node[fill=blue!100, draw=black, circle, minimum size=2.5mm] at (5,1){};

\node[] at (7.5,0.75) {$C_3$};
\node[fill=white!10, draw=black, circle, minimum size=2.5mm] (C31) at (7,1){};
\node[fill=white!10, draw=black, circle, minimum size=2.5mm] (C32) at (8,1){};
\node[fill=white!10, draw=black, circle, minimum size=2.5mm] (C33) at (9,1){};
\draw[very thick] (C31)--(C32);
\draw[very thick] (C32)--(C33);
\draw[very thick] (C31)--(7,0.5)--(9,0.5)--(C33);

\node[fill=red!100, draw=black, circle, minimum size=2.5mm] (C31) at (7,1){};
\node[fill=blue!100, draw=black, circle, minimum size=2.5mm] (C33) at (9,1){};

\filldraw[fill=red, fill opacity=0.3] (1,1.5)--(1,2.5)--(4.5,2.5)--(4.5,1.5)--(4,1.5)--(4,2)--(1.5,2)--(1.5,1.5)--(1,1.5);
\filldraw[fill=blue, fill opacity=0.3] (2,1.5)--(2,3.5)--(7.5,3.5)--(7.5,1.5)--(7,1.5)--(7,3)--(2.5,3)--(2.5,1.5)--(2,1.5);
\filldraw[fill=green, fill opacity=0.3]  (5,1.5)--(5,4.5)--(8.5,4.5)--(8.5,1.5)--(8,1.5)--(8,4)--(5.5,4)--(5.5,1.5)--(5,1.5);
\filldraw[fill=yellow, fill opacity=0.3] (3,1.5)--(3,5.5)--(9.5,5.5)--(9.5,1.5)--(9,1.5)--(9,5)--(6.5,5)--(6.5,1.5)--(6,1.5)--(6,5)--(3.5,5)--(3.5,1.5)--(3,1.5);


\node[] at (1.25,2.25) {$x_1$};
\node[fill=white!10, draw=black, circle, minimum size=2.5mm] (x1C1t) at (1,1.5){};
\node[fill=white!10, draw=black, circle, minimum size=2.5mm] (x1C1f) at (1.5,1.5){};
\node[fill=white!10, draw=black, circle, minimum size=2.5mm] (x1C2t) at (4,1.5){};
\node[fill=white!10, draw=black, circle, minimum size=2.5mm] (x1C2f) at (4.5,1.5){};

\node[fill=red!100, draw=black, circle, minimum size=2.5mm]  at (1,1.5){};
\node[fill=blue!100, draw=black, circle, minimum size=2.5mm]  at (1.5,1.5){};
\node[fill=red!100, draw=black, circle, minimum size=2.5mm]  at (4,1.5){};
\node[fill=blue!100, draw=black, circle, minimum size=2.5mm]  at (4.5,1.5){};

\draw[very thick] (x1C1t)--(1,2.5)--(4.5,2.5)--(x1C2f)--(x1C2t)--(4,2)--(1.5,2)--(x1C1f)--(x1C1t);

\node[] at (7.25,3.25) {$x_2$};
\node[fill=white!10, draw=black, circle, minimum size=2.5mm] (x2C1t) at (2,1.5){};
\node[fill=white!10, draw=black, circle, minimum size=2.5mm] (x2C1f) at (2.5,1.5){};
\node[fill=white!10, draw=black, circle, minimum size=2.5mm] (x2C3t) at (7,1.5){};
\node[fill=white!10, draw=black, circle, minimum size=2.5mm] (x2C3f) at (7.5,1.5){};

\node[fill=red!100, draw=black, circle, minimum size=2.5mm] at (2,1.5){};
\node[fill=blue!100, draw=black, circle, minimum size=2.5mm] at (2.5,1.5){};
\node[fill=red!100, draw=black, circle, minimum size=2.5mm] at (7,1.5){};
\node[fill=blue!1000, draw=black, circle, minimum size=2.5mm] at (7.5,1.5){};

\draw[very thick] (x2C1t)--(2,3.5)--(7.5,3.5)--(x2C3f)--(x2C3t)--(7,3)--(2.5,3)--(x2C1f)--(x2C1t);
\node[] at (5.25,4.25) {$x_3$};
\node[fill=white!10, draw=black, circle, minimum size=2.5mm] (x3C2t) at (5,1.5){};
\node[fill=white!10, draw=black, circle, minimum size=2.5mm] (x3C2f) at (5.5,1.5){};
\node[fill=white!10, draw=black, circle, minimum size=2.5mm] (x3C3t) at (8,1.5){};
\node[fill=white!10, draw=black, circle, minimum size=2.5mm] (x3C3f) at (8.5,1.5){};

\node[fill=red!100, draw=black, circle, minimum size=2.5mm] at (5,1.5){};
\node[fill=blue!100, draw=black, circle, minimum size=2.5mm] at (5.5,1.5){};
\node[fill=red!100, draw=black, circle, minimum size=2.5mm] at (8,1.5){};
\node[fill=blue!100, draw=black, circle, minimum size=2.5mm] at (8.5,1.5){};

\draw[very thick] (x3C2t)--(5,4.5)--(8.5,4.5)--(x3C3f)--(x3C3t)--(8,4)--(5.5,4)--(x3C2f)--(x3C2t);
\node[] at (6.25,5.25) {$x_4$};
\node[fill=white!10, draw=black, circle, minimum size=2.5mm] (x4C1t) at (3,1.5){};
\node[fill=white!10, draw=black, circle, minimum size=2.5mm] (x4C1f) at (3.5,1.5){};
\node[fill=white!10, draw=black, circle, minimum size=2.5mm] (x4C2t) at (6,1.5){};
\node[fill=white!10, draw=black, circle, minimum size=2.5mm] (x4C2f) at (6.5,1.5){};
\node[fill=white!10, draw=black, circle, minimum size=2.5mm] (x4C3t) at (9,1.5){};
\node[fill=white!10, draw=black, circle, minimum size=2.5mm] (x4C3f) at (9.5,1.5){};

\node[fill=red!100, draw=black, circle, minimum size=2.5mm] at (3,1.5){};
\node[fill=blue!100, draw=black, circle, minimum size=2.5mm] at (3.5,1.5){};
\node[fill=red!100, draw=black, circle, minimum size=2.5mm] at (6,1.5){};
\node[fill=blue!100, draw=black, circle, minimum size=2.5mm] at (6.5,1.5){};
\node[fill=red!100, draw=black, circle, minimum size=2.5mm] at (9,1.5){};
\node[fill=blue!100, draw=black, circle, minimum size=2.5mm] at (9.5,1.5){};

\draw[very thick] (x4C1t)--(3,5.5)--(9.5,5.5)--(x4C3f)--(x4C3t)--(9,5)--(6.5,5)--(x4C2f)--(x4C2t)--(6,5)--(3.5,5)--(x4C1f)--(x4C1t);

\draw[very thick] (C11)--(x1C1f);
\draw[very thick] (C12)--(x2C1t);
\draw[very thick] (C13)--(x4C1t);

\draw[very thick] (C21)--(x1C2t);
\draw[very thick] (C22)--(x3C2t);
\draw[very thick] (C23)--(x4C2f);

\draw[very thick] (C31)--(x2C3f);
\draw[very thick] (C32)--(x3C3f);
\draw[very thick] (C33)--(x4C3t);




\end{tikzpicture}
    \caption{
    \revcr{Another drawing of $G$ (shown in \Cref{fig:ISRmaxdg3}), corresponding to the Boolean formula $\phi'=(x_1\lor \overline{x_2}\lor\overline{x_4})\land(\overline{x_1}\lor \overline{x_3}\lor x_4)\land (x_2\lor x_3\lor \overline{x_4})$, on a grid with $11$ rows and $18$ columns. The intersection of a dotted horizontal line labeled $i \in \{1,\ldots,11\}$ and a dotted vertical line labeled $j \in \{1,\ldots,18\}$ represents the coordinate $(i,j)$. Thick lines represent the edges of the graph $G$, and all edge crossings occur only between edges belonging to distinct variable gadgets.}
    }
    \label{fig:GraphRedrawed}
\end{figure}
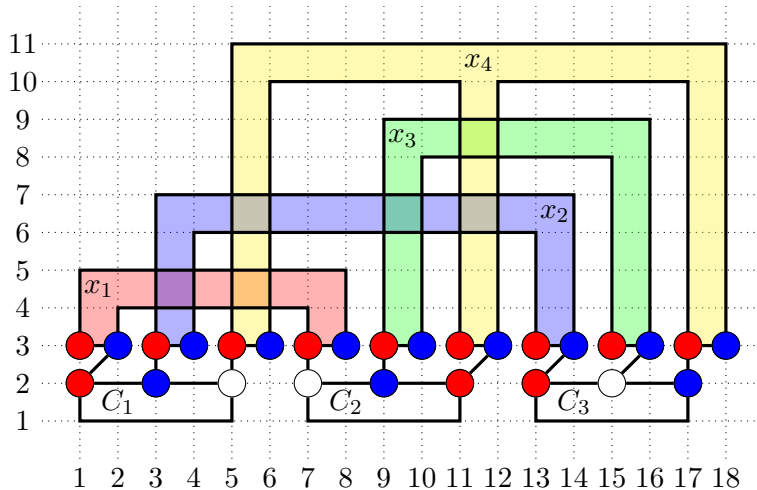
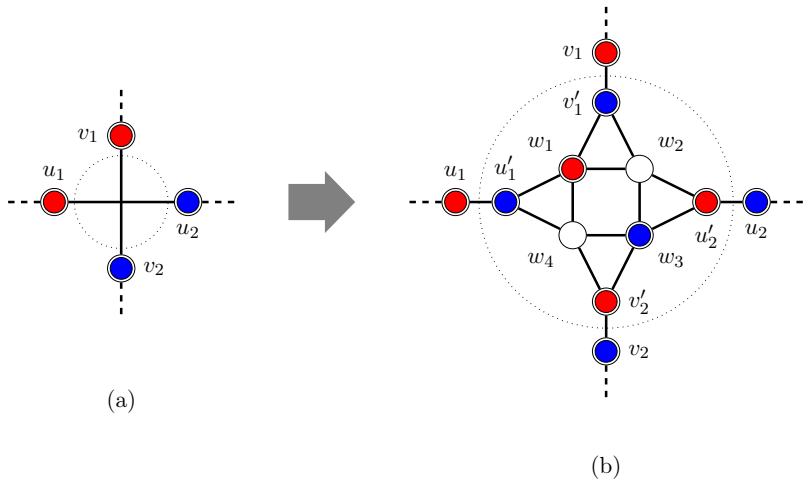
\begin{figure}[t]
    \centering
    \scalebox{0.8}{\begin{tikzpicture}[scale=1.1]
\node[] at (-3,-3) {(a)};

\node[fill=white!10, draw=black, circle, minimum size=4.5mm, label=90:\normalsize {$u_1$}] (Au1) at (-4,0){};
\node[fill=white!10, draw=black, circle, minimum size=4.5mm, label=270:\normalsize {$u_2$}] (Au2) at (-2,0){};
\node[fill=white!10, draw=black, circle, minimum size=4.5mm, label=180:\normalsize {$v_1$}] (Av1) at (-3,1){};
\node[fill=white!10, draw=black, circle, minimum size=4.5mm, label=0:\normalsize {$v_2$}] (Av2) at (-3,-1){};
\draw[very thick] (Au1)--(Au2);
\draw[very thick] (Av1)--(Av2);

\draw[very thick, dashed] (Au1)--(-4.75,0);
\draw[very thick, dashed] (Au2)--(-1.25,0);
\draw[very thick, dashed] (Av1)--(-3,1.75);
\draw[very thick, dashed] (Av2)--(-3,-1.75);

\node[fill=red!100, draw=black, circle, minimum size=2.5mm] (Au1token) at (-4,0){};
\node[fill=red!100, draw=black, circle, minimum size=2.5mm] (Av1token) at (-3,1){};

\node[fill=blue!100, draw=black, circle, minimum size=2.5mm] (Au2token) at (-2,0){};
\node[fill=blue!100, draw=black, circle, minimum size=2.5mm] (Av2token) at (-3,-1){};

\draw (-3,0) [dotted] circle [radius=0.7];

\node[] at (4.25,-4) {(b)};

\node[fill=white!10, draw=black, circle, minimum size=4.5mm, label=90:\normalsize {$u_1$}] (Out1) at (2,0){};
\node[fill=white!10, draw=black, circle, minimum size=4.5mm, label=270:\normalsize {$u_2$}] (Out2) at (6.5,0){};
\node[fill=white!10, draw=black, circle, minimum size=4.5mm, label=180:\normalsize {$v_1$}] (Out3) at (4.25,2.25){};
\node[fill=white!10, draw=black, circle, minimum size=4.5mm, label=0:\normalsize {$v_2$}] (Out4) at (4.25,-2.25){};

\node[fill=white!10, draw=black, circle, minimum size=4.5mm, label=90:\normalsize {$u'_1$}] (Mid1) at (2.75,0){};
\node[fill=white!10, draw=black, circle, minimum size=4.5mm, label=270:\normalsize {$u'_2$}] (Mid2) at (5.75,0){};
\node[fill=white!10, draw=black, circle, minimum size=4.5mm, label=180:\normalsize {$v'_1$}] (Mid3) at (4.25,1.5){};
\node[fill=white!10, draw=black, circle, minimum size=4.5mm, label=0:\normalsize {$v'_2$}] (Mid4) at (4.25,-1.5){};

\node[fill=white!10, draw=black, circle, minimum size=4.5mm, label=135:\normalsize {$w_1$}] (In1) at (3.75,0.5){};
\node[fill=white!10, draw=black, circle, minimum size=4.5mm, label=225:\normalsize {$w_4$}] (In2) at (3.75,-0.5){};
\node[fill=white!10, draw=black, circle, minimum size=4.5mm, label=45:\normalsize {$w_2$}] (In3) at (4.75,0.5){};
\node[fill=white!10, draw=black, circle, minimum size=4.5mm, label=315:\normalsize {$w_3$}] (In4) at (4.75,-0.5){};

\draw[very thick, dashed] (Out1)--(1.25,0);
\draw[very thick, dashed] (Out2)--(7.25,0);
\draw[very thick, dashed] (Out3)--(4.25,3);
\draw[very thick, dashed] (Out4)--(4.25,-3);

\draw[very thick] (Out1)--(Mid1);
\draw[very thick] (Out2)--(Mid2);
\draw[very thick] (Out3)--(Mid3);
\draw[very thick] (Out4)--(Mid4);

\draw[very thick] (Mid1)--(In1);
\draw[very thick] (Mid1)--(In2);
\draw[very thick] (Mid2)--(In3);
\draw[very thick] (Mid2)--(In4);
\draw[very thick] (Mid3)--(In1);
\draw[very thick] (Mid3)--(In3);
\draw[very thick] (Mid4)--(In2);
\draw[very thick] (Mid4)--(In4);

\draw[very thick] (In1)--(In3);
\draw[very thick] (In2)--(In4);
\draw[very thick] (In4)--(In3);
\draw[very thick] (In1)--(In2);

\node[fill=red!100, draw=black, circle, minimum size=2.5mm] (Out1redtoken) at (2,0){};
\node[fill=red!100, draw=black, circle, minimum size=2.5mm] (Out3redtoken) at (4.25,2.25){};
\node[fill=red!100, draw=black, circle, minimum size=2.5mm] (In1redtoken) at (3.75,0.5){};
\node[fill=red!100, draw=black, circle, minimum size=2.5mm] (Mid2redtoken) at (5.75,0){};
\node[fill=red!100, draw=black, circle, minimum size=2.5mm] (Mid4redtoken) at (4.25,-1.5){};

\node[fill=blue!100, draw=black, circle, minimum size=2.5mm] (Out1bluetoken) at (2.75,0){};
\node[fill=blue!100, draw=black, circle, minimum size=2.5mm] (Out3bluetoken) at (4.25,1.5){};
\node[fill=blue!100, draw=black, circle, minimum size=2.5mm] (In1bluetoken) at (4.75,-0.5){};
\node[fill=blue!100, draw=black, circle, minimum size=2.5mm] (Mid2bluetoken) at (6.5,0){};
\node[fill=blue!100, draw=black, circle, minimum size=2.5mm] (Mid4bluetoken) at (4.25,-2.25){};
\draw [->, >={Triangle[width=10mm, length=5mm]}, line width=6mm, gray] (-0.5,0) -- (0.5,0);

\draw (4.25,0) [dotted] circle [radius=1.9];






\end{tikzpicture}}
    \caption{
    \revcr{
    An illustration of adding tokens to $I$ and $J$, yielding $I^*$ and $J^*$. (a) Vertices $u_1$ and $v_1$ belong to $I$ (red tokens), and $u_2$ and $v_2$ belong to $J$ (blue tokens). (b) In $I^*$, the new vertices $w_1$, $u'_2$, and $v'_2$ are added, while in $J^*$, the new vertices $w_3$, $u'_1$, and $v'_1$ are added. The two independent sets $I^*$ and $J^*$ remain disjoint.
    }
    }
    \label{fig:CrossGadgetIniTar}
\end{figure}

\begin{figure}[t]
    \centering
    \scalebox{0.8}{\input{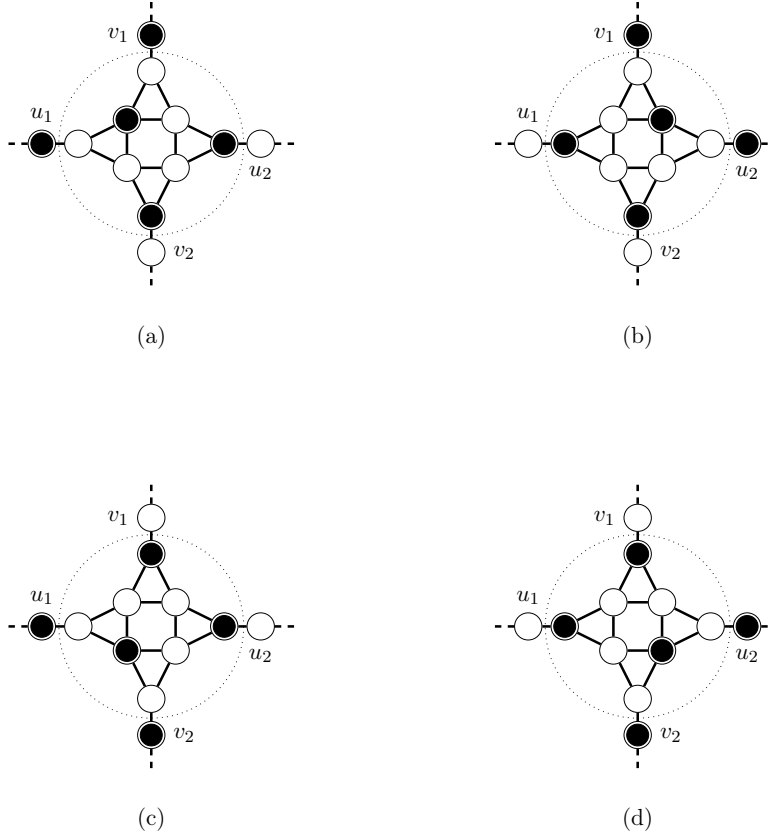}}
    \caption{
    \revcr{An illustration of how the crossover gadget works. The crossing edges $u_1u_2$ and $v_1v_2$, where one endpoint of each edge is occupied by a token, are replaced by a crossover gadget:
    (a) $u_1$ and $v_1$ are occupied; (b) $u_2$ and $v_1$ are occupied; (c) $u_1$ and $v_2$ are occupied; (b) $u_2$ and $v_2$ are occupied.
    }
    }
    \label{fig:HowWorkCrossGadget}
\end{figure}

\begin{theorem}\label{ISR_NPcom_planarmaxdeg4}
    Let $\mu$ be any fixed positive integer. \prb{ISR} under $\kTJ$ is $\NP$-hard for planar graphs $G$ of maximum degree $4$ when $k=|I|-\mu\geq 1$, where $I$ is an initial independent set of $G$.
\end{theorem}
\begin{proof}
\revcr{As the same reason in the proof of \Cref{ISR_NPcom_maxdeg3}, we provide the proof when $\mu=1$ (that is, $k=|I|-1$).}
Let \((G, I, J, \kTJ)\) be the instance of \prb{ISR} constructed from \(\phi'\) in the proof of \Cref{ISR_NPcom_maxdeg3}, \revcr{where $k=|I|-1$}.  
Consider any drawing of $G$ on the plane, which may have crossings.

To make $G$ into a planar graph of maximum degree $4$, we replace each crossing in the drawing with a crossover gadget, which is shown in \Cref{fig:CrossGadget}.
Suppose that we replaced a crossing of edges \(u_1u_2\) and \(v_1v_2\) with a crossover gadget.
We claim that the crossover gadget correctly ``simulates'' the crossing if the crossing is \emph{good}, that is, every independent set \(I'\) of $G$ with size \(|I|\) satisfies \(|I' \cap \{u_1, u_2\}| = |I' \cap \{v_1, v_2\}| = 1\).
Although not every drawing of $G$ meets this condition\footnote{For example, the drawing in \Cref{fig:ISRmaxdg3} contains a non-good crossing of two edges \revcr{$f^2_{x_3}v_{x_3}$ and $f^2_{x_4}u_{x_4}$, where the two vertices $v_{x_3}$ and $u_{x_4}$ correspond to the literal $x_3$ in clause $C_3$ and the literal $x_4$ in clause $C_2$, respectively. In this crossing, $f^2_{x_3}$ and $v_{x_3}$ are both not in the independent set $J$.}}, we claim that $G$ always admits a drawing where all crossings are good.
To this end, we construct a drawing of $G$ in which all edge crossings occur only between edges belonging to distinct variable gadgets.

Let $m$ be the number of clauses and $n$ be the number of variables in $\phi'$, respectively.
We will represent all variable gadgets and clause gadgets on a grid with $(2n+3)$ rows and $6m$ columns such that edges between variable gadgets and clause gadgets have no crossing (see also \Cref{fig:GraphRedrawed}).
Let $(i,j)$ be the coordinates of a point on the Euclidean plane, where $i \in [6m]$ and $j \in [2n+3]$.
For each clause gadget corresponding to a clause \(C_i\) of \(\phi'\), the three literal vertices of the gadget are positioned at \((6(i-1)+1,2)\), \((6(i-1)+3,2)\), and \((6(i-1)+5,2)\).
Furthermore, the edges are embedded to minimize the sum of their length, with one of the three edges utilizing the bottom border.
\revcr{
If a literal vertex, corresponding to the $j$-th occurrence of a literal of variable $x$, is placed at $(i',2)$, then the true vertex $t^j_x$ and the false vertex $f^j_x$ are positioned at $(i',3)$ and $(i'+1,3)$, respectively, and are connected by a shortest edge.
}

Additionally, each literal vertex is joined with either of them, which follows the construction of $G$ in the proof of \Cref{ISR_NPcom_maxdeg3}: \revcr{for a vertex $v$ in a clause gadget corresponding to a literal $l$ of a variable $x$ with occurrence $j \in [a_x]$, connect $v$ to $t_x^{j}$ if $l$ is negative, and to $f_x^{j}$ otherwise.}
Then, for a variable gadget corresponding to a variable \(x_p\) with \(p \in [n]\), we embed the remaining edges in the gadget using only the vertical lines where its vertices are located and the \((2p+2)\)-th and \((2p+3)\)-th horizontal lines, minimizing the total length. (If the variable gadget is a path with 2 vertices, there is nothing to do.)

Now, a new embedding of $G$ is obtained and denoted by $D(G)$.
Although $D(G)$ and $G$ are isomorphic, $D(G)$ only has crossing edges in variable gadgets.
\revcr{Since every variable gadget forms a cycle of even length and has tokens placed alternately on its vertices, exactly one endpoint of each edge in the variable gadgets belongs to any independent set of $G$ of size exactly $|I|$.}
Thus, all crossings of $D(G)$ are good.

We pick a crossing of two edges, say $u_1u_2$ and $v_1v_2$. 
Since both $u_1u_2$ and $v_1v_2$ are edges of variable gadgets, we may assume without loss of generality that $u_1, v_1 \in I$ and $u_2, v_2 \in J$. 
We then replace this crossing with our crossover gadget.
Let $G^*$ denote the resulting graph. 
\revcr{
For the initial (resp.\ target) independent set $I$ (resp.\ $J$) of $G$, adding three tokens on $w_1$, $u'_2$, and $v'_2$ (resp.\ $w_3$, $u'_1$, and $v'_1$) in the crossover gadget (as shown in \Cref{fig:CrossGadgetIniTar}) results in the independent set $I^*$ (resp.\ $J^*$).
}
Let $(G^*, I^*, J^*, \kTJ)$ be a new instance of \prb{ISR} under $\kTJ$, where $k=|I^*|-1$.
\revcr{Note that, since the sets of vertices added to $I$ and $J$ are disjoint, we have $I^* \cap J^* = \emptyset$.
}

\revcr{For each configuration of tokens on $u_1$, $u_2$, $v_1$, and $v_2$ in $G$, there exists a unique configuration of three tokens on the crossover gadgets in $G^*$ as shown in \Cref{fig:HowWorkCrossGadget}.
Thus, the configuration of tokens on the vertices of the crossover gadget can be changed in $G^*$ if and only if the configuration of tokens on $u_1$, $u_2$, $v_1$, and $v_2$ is changed in $G$.
}
Therefore, by faithfully simulating token moves along the edges $u_1u_2$ and $v_1v_2$ of $G$ using the token arrangements shown in \Cref{fig:HowWorkCrossGadget}, one can show that $(G, I, J, \kTJ)$ is a yes-instance if and only if $(G^*, I^*, J^*, \kTJ)$ is a yes-instance.

The graph $G^*$ has one less edge crossing than $G$.
Consequently, we can obtain a plane graph $G'$ by repeating the above replacement for each crossing of $D(G)$.
As $D(G)$ has $O(mn)$ crossings, the construction of $G'$ can be done in polynomial time.
Since a crossover gadget is a graph of maximum degree $4$, $G'$ is a plane graph of maximum degree $4$.
This completes our proof.
\end{proof}

\paragraph{Proof of \Cref{IntSAT_NPcomp}}\label{appendixsubsec:IntSAT_NPcomp}
\begin{proof}
    It is obvious that \prb{IntE3-SAT} is in $\NP$.
    To prove the $\NP$-completeness, we use a polynomial-time reduction from \prb{E3-SAT}.
    
    Let $\phi=C_1\land C_2\land \cdots \land C_m$ be an instance of \prb{E3-SAT}, and $X=\{x_1,\ldots,x_n\}$ be the variable set of $\phi$. We will construct a sandwiched E3-CNF formula $\phi'$ with $7mn^2$ clauses and $n+2mn^2$ variables at the end of our construction.
    Before starting our construction, we restrict $\phi$ so that neither the all-$\True$ assignment nor the all-$\False$ assignment satisfies $\phi$.
    \prb{E3-SAT} remains $\NP$-hard with this restriction; otherwise, we could solve \prb{E3-SAT} without the restriction by checking whether the all-$\True$ or all-$\False$ assignment satisfies $\phi$ at first.
    We obtain a new CNF formula $\phi^*$ as follows:
    \begin{align*}
        \phi^*&=\phi \lor \bigwedge_{i=1}^{n}x_i \lor \bigwedge_{i=1}^{n}\overline{x_i}=\bigwedge_{h=1}^{m}C_h \lor \bigwedge_{i=1}^{n}x_i \lor \bigwedge_{i=1}^{n}\overline{x_i}
        =\bigwedge_{h=1}^{m}\bigwedge_{i=1}^{n}\bigwedge_{j=1}^{n}(C_h\lor x_i\lor \overline{x_j}).
    \end{align*}
    Since $C_h$ consists of exactly three literals, $C_h\lor x_i \lor \overline{x_j}$ is a disjunction of five literals. Thus, $\phi^*$ is an E5-CNF formula. Furthermore, $\phi^*$ is a sandwiched E5-CNF formula, that is, each clause of $\phi^*$ has both a positive literal $x_i$ and a negative literal $\overline{x_j}$.
    From our conversion, $\phi^*$ has $mn^2$ clauses and $n$ variables.
    It is observed that for any mixed variable assignment $b$, $\phi$ evaluates $\True$ if and only if $\phi^*$ evaluates $\True$.
    Thus, we can immediately say that $\phi$ is a yes-instance of \prb{E3-SAT} without all-$\True$ and all-$\False$ assignments if and only if $\phi'$ is a yes-instance of \prb{IntE5-SAT}.

    Then, we explain how to convert a sandwiched E5-CNF formula $\phi^*$ to a sandwiched E3-CNF formula $\phi'$ with $7mn^2$ clauses and $n+2mn^2$ variables.
    We repeat the following operation \revcr{on} $\phi^*$ until all clauses have size exactly $3$.
    Let $\phi^*_0 = \phi^*$ and $\phi^*_i$ for a positive integer $i$ be a formula obtained from a sandwiched CNF formula $\phi^*_{i-1}$.
    
    Let $C$ be a clause in $\phi^*_{i-1}$ with size at least~$4$.
    Then $C$ contains $x \lor y$ or $\overline{x} \lor \overline{y}$ for variables $x,y$.
    If $C$ contains $x \lor y$, then we replace $x \lor y$ with a new variable $z$ and combine the modified formula with the formula $x \lor y \leftrightarrow z$ using the $\land$ operator.
    Furthermore, $x\lor y\leftrightarrow z$ is transformed as follows:
    \begin{align}
        x \lor y \leftrightarrow z
        &= (x\lor y \lor \overline{z})\land (\overline{(x\lor y)} \lor z)\notag\\
        &= (x\lor y \lor \overline{z})\land ((\overline{x}\land \overline{y}) \lor z)\notag\\
        &= (x\lor y \lor \overline{z})\land (\overline{x}\lor z)\land (\overline{y}\lor z)\notag\\
        &= (x\lor y \lor \overline{z})\land (\overline{x}\lor \overline{x}\lor z)\land (\overline{y}\lor \overline{y}\lor z)\label{equationtocnf}.
    \end{align}

    We claim that $\phi^*_i$ obtained from $\phi^*_{i-1}$ by the above operation is a sandwiched CNF formula. 
    Clearly, each clause in $\phi^*_i$ that remains unchanged from $\phi^*_{i-1}$ contains both positive and negative literals.
    The clause in $\phi^*_i$ obtained from $C$ by replacing $x \lor y$ with a positive literal $z$ contains a negative literal because $C$ also has a negative literal.
    Combined with \Cref{equationtocnf}, $\phi^*_i$ is a sandwiched CNF formula.
    Consequently, if there is a clause in $\phi^*_{i-1}$ that contains $x \lor y$, then a sandwiched CNF formula $\phi^*_i$ is obtained.
    
    Similarly, if $C$ contains $\overline{x}\lor \overline{y}$, then we replace $\overline{x}\lor \overline{y}$ with the negative literal of a new variable $z$ and combine the modified formula with the formula $\overline{x}\lor \overline{y}\leftrightarrow \overline{z} =(\overline{x}\lor\overline{y}\lor z)\land(x\lor x\lor \overline{z})\land (y\lor y\lor \overline{z})$ using the $\land$ operator.
    As with the previous argument, we can say that $\phi^*_i$ is a sandwiched CNF formula.

    Let $\phi'$ be the sandwiched CNF formula obtained from $\phi^* = \phi^*_0$ by repeating the above operation until all clauses have size exactly $3$.
    Since exactly two replacements occur per clause in $\phi^*$, we have $\phi' = \phi^*_{2mn^2}$.
    Moreover, six new clauses and two new variables are added for each clause in $\phi^*$.
    Therefore, $\phi'$ has $7mn^2$ clauses and $n+2mn^2$ variables.

    To complete our reduction, for each $i\in [2mn^2]$, we will show that there exists a mixed variable assignment that satisfies $\phi^*_{i-1}$ if and only if there exists a mixed variable assignment that satisfies $\phi^*_{i}$. 
    This immediately implies that $\phi^* = \phi^*_0$ is a yes-instance of \prb{IntE5-SAT} if and only if $\phi' = \phi^*_{2mn^2}$ is a yes-instance of \prb{IntE3-SAT}.
    Let $X_i$ be the set of variables in $\phi^*_{i}$.
    
    Suppose that there is a mixed variable assignment $b_{i-1}$ for $X_{i-1}$ that satisfies $\phi^*_{i-1}$. 
    Consider a variable assignment $b_{i}$ for $X_i = X_{i-1} \cup \{ z\}$ such that $b_{i}(x)=b_{i-1}(x)$ for each $x\in X_{i-1}$.
    Moreover, set $b_i(z)= x \lor y$ if $z$ replaces $x\lor y$, and set $b_i(z)= \overline{\overline{x} \lor \overline{y}}$ if $\overline{z}$ replaces $\overline{x} \lor \overline{y}$.
    Since $b_{i-1}$ is a mixed variable assignment, $b_i$ is also a mixed variable assignment.
    Furthermore, it is easy to see that $b_i$ satisfies $\phi^*_i$.

    Conversely, suppose that there is a mixed variable assignment $b_i$ for $X_i$ that satisfies $\phi^*_i$. 
    Due to $x \lor y \leftrightarrow z$ or $\overline{x} \lor \overline{y} \leftrightarrow \overline{z}$, the variable assignment $b_{i-1}$ such that $b_{i-1}(x)=b_i(x)$ for each $x\in X_{i-1}$ satisfies $\phi^*_{i-1}$.
    We claim that the variable assignment $b_{i-1}$ is mixed.
    For the sake of contradiction, assume that $b_{i-1}$ is not mixed, that is, $b_{i-1}$ is either the all-$\True$ assignment or all-$\False$ assignment. 
    If $x \lor y \leftrightarrow z$ is added into $\phi^*_{i-1}$, then $b_i(x) = b_i(y) = b_i(z)$ holds.
    Thus, $b_{i}$ is also not mixed, a contradiction.
    Similarly, if $\overline{x} \lor \overline{y} \leftrightarrow \overline{z}$ is added into $\phi^*_{i-1}$, then $b_i(x) = b_i(y) = b_i(z)$ holds, a contradiction.
    Therefore, we conclude that $b_{i-1}$ is a mixed variable assignment for $X_{i-1}$ that satisfies $\phi^*_{i-1}$.
    This completes the proof.
\end{proof}

\subsubsection{Membership in NP}\label{subsubsec:ISRmemberNP}
Next, we show that \prb{ISR} under $\kTJ$ with $k=|I|-\mu$ is in $\NP$ not only when $\mu$ is constant but also $\mu$ is at most $O(\log|I|)$ for graphs of bounded maximum degree and planar graphs of maximum degree $o(\frac{n}{\log n})$, where $n$ is the number of vertices in the input graph. 

\begin{theorem}\label{ISRinNPlogmu}
    Let $G$ be an input graph with $n$ vertices, chromatic number $O(1)$ and maximum degree $o(\frac{n}{\log n})$, and let $I$ be an initial independent set of $G$.
    \prb{ISR} under $\kTJ$ is in $\NP$ when $k=|I|-\mu\geq 1$ with any non-negative integer $\mu$ at most $O(\log|I|)$.
\end{theorem}
To prove \Cref{ISRinNPlogmu}, we will evaluate the length of a shortest reconfiguration sequence between any two independent sets of an input graph.
Firstly, we introduce a lemma on the maximum size of intersecting families of sets.
Let $N,r$ be positive integers with $N\geq r$, and let $L\subseteq \{0,1,\ldots,r-1\}$.
We say that a family $\mathcal{F}$ of $r$-element subsets of $[N]$ is an \emph{$(N,r,L)$-system} if $|F\cap F'|\in L$ holds for all distinct $F,F'\in \mathcal{F}$.
Let $m(N,r,L)$ denote the maximum size of $(N,r,L)$-systems.

For any positive integers $N,r,p$ with $N\geq r > p \geq 1$, the following upper bound is known~(see, for example, \cite{DBLP:FranklT16a,GyulaTM64, DBLP:Rodl85}):
\begin{align}\label{align:upper_tau}
    m(N,r,\{0,\ldots,p-1\}) \leq \tbinom{N}{p}/\tbinom{r}{p} .
\end{align}

Using \Cref{align:upper_tau}, we prove \Cref{inNPfortheta}, which provides an upper bound on the length of any shortest reconfiguration sequence in the general case.

\begin{lemma}\label{inNPfortheta}
    Let $G$ be an input graph with $n=|V(G)|$, $I$ and $J$ be initial and target independent sets of $G$, and $\mu<|I|$ be any non-negative integer.
    If $I$ and $J$ are reconfigurable under $\kTJ$, then the length of a shortest reconfiguration sequence between $I$ and $J$ under $\kTJ$, where $k=|I|-\mu \geq 1$, is at most $O((\frac{n}{k})^\mu)$.
\end{lemma}
\begin{proof}
    When $\mu=0$, there exists a reconfiguration sequence $\langle I,J\rangle$ of length $1$. 
    Therefore, we assume that $\mu\geq 1$.
    Consider any shortest reconfiguration sequence $\sigma=\langle I=I_0, I_1,\ldots,I_\ell = J \rangle$ between $I$ and $J$.
    Let $\mathcal{I}=\{I_i \colon i\in\{0,1,\ldots,\ell\}, i \text{ is even}\}$.
    Note that $|\mathcal{I}|=\lfloor \frac{\ell}{2}\rfloor+1$.
    Since $\sigma$ is the shortest reconfiguration sequence, for any two independent sets $I_i$ and $I_j$ in $\mathcal{I}$ with $i < j$, we have $|I_i\cap I_j|<\mu$; otherwise, $I_i$ and $I_j$ are adjacent under $\kTJ$ and we can obtain a shorter reconfiguration sequence $\sigma'$ from $\sigma$ by removing all independent sets $I_{i+1},\ldots,I_{j-1}$, which is a contradiction.
    Then, we observe that $\mathcal{I}$ is an $(n,|I|,\{0,\ldots,\mu-1\})$-system.
    By using \Cref{align:upper_tau}, we have
    \begin{align*}
        |\mathcal{I}|=\biggr\lfloor\frac{\ell}{2}\biggr\rfloor+1 \leq m(n,|I|,\{0,\ldots,\mu-1\}) \leq \tbinom{n}{\mu}/\tbinom{|I|}{\mu} \leq \biggl(\frac{n}{|I|-\mu}\biggr)^{\mu}.
    \end{align*}
    Therefore, the length $\ell$ of $\sigma$ is at most $O((\frac{n}{|I|-\mu})^\mu)=O((\frac{n}{k})^\mu)$.
\end{proof}

Now, we can prove \Cref{ISRinNPlogmu}.

\begin{proof}[Proof of \Cref{ISRinNPlogmu}]
    It is trivial when $\mu=0$, thus we assume that $\mu\geq 1$.
    Suppose first that $2\mu\geq|I|$. 
    Then, there is some constant $c$ such that $|I| < c$ since $\mu=O(\log |I|)$. 
    We can solve \prb{ISR} under $\kTJ$ by enumerating all independent sets of $G$ with constant size in polynomial time.
    
    Suppose next that $2\mu<|I|$ for sufficiently large $|I|$.
    Let $A$ and $B$ be two independent sets such that $A\subseteq I$ and $B\subseteq J$ with size exactly $\mu$. 
    Let $G^*$ be the subgraph of $G$ obtained by removing all vertices in $N[A\cup B]$.
    If $G^*$ has an independent set $I^*$ with size $k$, then there is a reconfiguration sequence $\langle I, A\cup I^*, B\cup I^*, J \rangle$ between $I$ and $J$.
    We say that this reconfiguration sequence is a \emph{simple reconfiguration sequence}.
    Since the vertex set of $G^*$ is $V(G)\setminus N[A\cup B]$, the size of $V(G^*)$ is at least $n-2{\mu}(\Delta+1)$, where $\Delta$ is the maximum degree of $G$.
    Let $\chi$ and $\chi^*$ be the chromatic numbers of $G$ and $G^*$, respectively.
    Then, we observe that $\chi^*\leq \chi$.
    By the relationship between the chromatic number and the independence number, $G^*$ has an independent set with size at least $|V(G^*)|/\chi^*\geq(n-2{\mu} (\Delta+1))/{\chi^*}$.
    Thus, if \((n - 2\mu (\Delta + 1)) / \chi^* \geq  k\), then \(I\) and \(J\) are always reconfigurable, as a simple reconfiguration sequence exists.
    Note that the length of a simple reconfiguration sequence is $3$.

    It remains to consider the case where \((n - 2\mu (\Delta + 1)) / \chi^* < k\), in which a simple reconfiguration sequence between \(I\) and \(J\) may not exist.
    Combined with \Cref{inNPfortheta}, the length $\ell$ of a shortest reconfiguration sequence between $I$ and $J$ satisfies
    \begin{align}\label{master}
        \ell =O\biggr(\biggr(\frac{n}{k}\biggl)^\mu\biggl) =  O\biggr(\biggr(\frac{n\chi^*}{n - 2\mu (\Delta + 1)}\biggl)^\mu\biggl) = O\biggr((\chi^*)^\mu
        \biggr(\frac{1}{1-\frac{2{\mu} (\Delta+1)}{n}}\biggl)^\mu\biggl).
    \end{align}
    Since $\mu=O(\log |I|)=O(\log n)$ and $\Delta=o(\frac{n}{\log n})$, we have $\mu\Delta=o(n)$.
    Hence, we have $(2\mu(\Delta+1))/n=o(1)$.
    Furthermore, $\chi^*\leq \chi=O(1)$.
    Thus, from \Cref{master}, we have $$\ell=O((\chi^*)^\mu(\frac{1}{1-o(1)})^\mu)=O(1)^{O(\log n)}=O(n^{O(1)}).$$
    Therefore, $\ell$ is polynomially bounded in $n$, and hence the problem belongs to $\NP$.
    This completes the proof.
\end{proof}

It is known that the chromatic number of $G$ is at most $\Delta + 1$, where $\Delta$ is the maximum degree of $G$~\cite{Brooks_1941}.
In addition, the chromatic number of any planar graph is at most 4~\cite{10.1215/ijm/1256049011,10.1215/ijm/1256049012}.
Therefore, \Cref{ISRinNPlogmu} gives the results including graphs of bounded maximum degree and planar graphs of maximum degree $o(\frac{n}{\log n})$.

\subsection{VCR}\label{VCRguaranteed}
In \Cref{ISRguaranteed}, we showed that \prb{ISR} under $\kTJ$ with $k = |I| - \mu$, where $\mu$ is any fixed positive integer, is $\NP$-complete even for graphs of maximum degree 3 and for planar graphs of maximum degree 4.
In contrast to this intractability, \prb{VCR} under $\kTJ$ is in {\XP} for general graphs when parameterized by $\mu = |S|-k>0$, where $S$ is an initial vertex cover of an input graph.

\begin{theorem}\label{VCR_XP}
    \prb{VCR} under $\kTJ$ is in {\XP} for general graphs $G$ when parameterized by $\mu=|S|-k\geq 0$, where $S$ is an initial vertex cover of $G$.
\end{theorem}

In the proof of \Cref{VCR_XP}, we present an \(\XP\) algorithm for the problem.

Let \((G, S, T, \kTJ)\) be an instance of \prb{VCR}. 
We consider the \emph{reconfiguration graph} \(\mathcal{C} = (\mathcal{V}, \mathcal{E})\) for the instance.
In a reconfiguration graph $\mathcal{C=(V,E)}$ of $G$, each vertex in $\mathcal{V}$ corresponds to a vertex cover $S'$ of $G$ with size exactly $|S|$. 
We call each vertex of $\mathcal{C}$ a \emph{node} in order to distinguish it from a vertex of $G$.
We use $w_{S'}$ to denote a node corresponding to a vertex cover $S'$.
Then, two nodes $w_{S'}$ and $w_{T'}$ are joined by an edge in $\mathcal{C}$ if and only if the corresponding two vertex covers $S'$ and $T'$ of $G$ with size exactly $|S|$ are adjacent under $\kTJ$. 
Since the number of such vertex covers can be superpolynomial, explicitly constructing \(\mathcal{C}\) in polynomial time may be infeasible, in general.

Our approach builds on the clique-compressed reconfiguration graph technique introduced in~\cite{ISIsoRsuga25}, which compactly represents \(\mathcal{C}\) by grouping cliques into single nodes.  
This compressed graph has at most \(O(n^\mu)\) nodes and preserves essential connectivity, making it sufficient for solving the problem if it is constructed efficiently.

Now, let us define a compressed version of a reconfiguration graph, \emph{clique-compressed reconfiguration graph} $\mathcal{C'=(V',E')}$~\cite{ISIsoRsuga25}.
We call each vertex of $\mathcal{C'}$ a \emph{clique-node} in order to distinguish it from a node of $\mathcal{C}$ and a vertex of $G$. 
There is a one-on-one correspondence between clique-nodes and vertex subsets of $V(G)$ with size exactly $\mu$. 
We use $w'_{X}$ to denote the clique-node assigned a vertex set $X\subseteq V(G)$ with size exactly $\mu$. 
Then two clique-nodes $w'_{X}$ and $w'_{Y}$ in $\mathcal{C'}$ are joined by an edge if and only if there exists a vertex cover $S'$ in $G$ with size exactly $|S'|$ such that $X\cup Y\subseteq S'$.

The following \Cref{claim:clique_nodes} asserts that the reconfiguration graph and its corresponding clique-compressed reconfiguration graph are essentially equivalent in terms of solution reachability.
\begin{lemma}\label{claim:clique_nodes}
    Let $(G, S, T, \kTJ)$ be an instance of $\prb{VCR}$, where $k = |S| - \mu$. 
    Consider the reconfiguration graph $\mathcal{C}$ and the clique-compressed reconfiguration graph $\mathcal{C'}$ for this instance. 
    Let $w_S$ and $w_T$ denote the nodes in $\mathcal{C}$ corresponding to $S$ and $T$, respectively, and let $w'_X$ and $w'_Y$ denote the clique-nodes in $\mathcal{C'}$ corresponding to subsets $X \subseteq S$ and $Y \subseteq T$ of size exactly $\mu$, respectively. 
    Then, there exists a path in $\mathcal{C}$ connecting $w_S$ and $w_T$ if and only if there exists a path in $\mathcal{C'}$ connecting $w'_X$ and $w'_Y$.
\end{lemma}
\begin{proof}
    First, we show the only-if direction. Suppose that there exists a path $P_{\mathcal{C}}=\langle w_{S}=w_{S_0},w_{S_1},\ldots,w_{S_\ell}=w_{T} \rangle_{\mathcal{C}}$ in $\mathcal{C}$ connecting $w_{S}$ and $w_{T}$, where $\ell$ is the length of $P_{\mathcal{C}}$. 
    Consider two vertex sets $S_{i-1}$ and $S_i$ for $i\in [\ell]$. Since $S_{i-1}$ and $S_i$ are adjacent, we have $|S_{i-1} \setminus S_i| \le k = |S| - \mu$, which implies $|S_{i-1}\cap S_i| \ge \mu$. 
    For each $i\in [\ell]$, let $X_i$ be a vertex set such that $X_i\subseteq S_{i-1}\cap S_i$ with size exactly $\mu$.
    For each $i \in [\ell-1]$, since $X_{i}\cup X_{i+1}\subseteq S_{i}$, the two corresponding clique-nodes $w'_{X_{i}}$ and $w'_{X_{i+1}}$ are joined by an edge in $\mathcal{C'}$. 
    Thus, there exists a path $P_{\mathcal{C'}}$ in $\mathcal{C'}$ connecting $w'_{X_1}$ and $w'_{X_\ell}$, where $X_1\subseteq S_0 = S$ and $X_\ell \subseteq S_\ell = T$. 
    If $X_1=X$ and $X_\ell=Y$, then we have the path between $w'_{X}$ and $w'_{Y}$ in $\mathcal{C'}$.
    Otherwise, consider $X_1 \neq X$ or $X_\ell \neq Y$.
    Since $X_1$ and $X$ (resp.\ $X_\ell$ and $Y$) are contained in $S$ (resp.\ $T$), $w'_{X_1}$ and $w'_{X}$ (resp.\ $w'_{X_\ell}$ and $w'_{Y}$) are joined by an edge in $\mathcal{C'}$.
    Thus, we have the path between $w'_{X}$ and $w'_{Y}$ in $\mathcal{C'}$.
    Therefore, there is a path connecting two clique-nodes $w'_{X}$ and $w'_{Y}$ in $\mathcal{C'}$.
		
	Next, we show the if direction, and hence suppose that there exists a path $P_{\mathcal{C'}}=\langle w'_X=w'_{X_0},w'_{X_1},\ldots,w'_{X_{\ell}}=w'_{Y} \rangle_{\mathcal{C'}}$ in $\mathcal{C'}$ connecting two clique-nodes $w'_{X}$ and $w'_{Y}$. 
    For each $i\in[\ell]$, let $S_{i-1}$ be a vertex cover such that $X_{i-1}\cup X_i\subseteq S_{i-1}$. 
    For each $i\in[\ell-1]$, since $|S_{i-1}\cap S_{i}|\geq|X_{i}|=\mu$, $S_{i-1}$ and $S_i$ are adjacent under $\kTJ$.
    Thus two nodes $w_{S_{i-1}}$ and $w_{S_{i}}$ are joined by an edge in $\mathcal{C}$. 
    Hence, there exists a path connecting $w_{S_0}$ and $w_{S_{\ell-1}}$. 
    If $S_0=S$ and $S_{\ell-1}=T$, we have the path between $w_S$ and $w_T$ in $\mathcal{C}$.
    Otherwise, consider $S_0\neq S$ or $S_{\ell-1}\neq T$.
    Since $S_0$ and $S$ (resp.\ $S_{\ell-1}$ and $T$) contain $X_0=X$ (resp.\ $X_\ell=Y$), $w_{S_0}$ and $w_S$ (resp.\ $w_{S_{\ell}}$ and $w_T$) are joined by an edge in $\mathcal{C}$.
    Thus, we have the path between $w_S$ and $w_T$ in $\mathcal{C}$, as claimed.
\end{proof}

Now, we prove \Cref{VCR_XP}.
\begin{proof}[Proof of \Cref{VCR_XP}]
    We give an {\XP} algorithm that solves the problem.
    If $|S\cap T|\geq \mu$, then the algorithm immediately returns the output YES, since $S$ and $T$ are adjacent under $\kTJ$. 
    Therefore, we consider an instance with $|S\cap T|< \mu$.
    
    The basic technique is the same as that provided in \cite{ISIsoRsuga25}.
    In our algorithm, we first construct a clique-compressed reconfiguration graph for a given instance, and then run the breadth-first search algorithm on the graph to determine whether the initial configuration and the target one are reconfigurable or not. 
    Let $(G,S,T,\kTJ)$ be an instance of \prb{VCR} under $\kTJ$ parameterized by $\mu=|S|-k$.
    We explain how to construct the clique-compressed reconfiguration graph $\mathcal{C'=(V',E')}$ without knowing the reconfiguration graph $\mathcal{C=(V,E)}$.
    
    For each vertex subset $X$ of $G$ with size exactly $\mu$, we create a clique-node in $\mathcal{C'}$ to which $X$ is assigned.
    We then construct the edge set $\mathcal{E'}$ of $\mathcal{C'}$. 
    Consider two clique-nodes $w'_{X}$ and $w'_{Y}$, where $X$ and $Y$ are vertex subsets of $G$ with $|X|=|Y|=\mu$.
    Let $Z = X\cup Y$.
    To check the existence of an edge between $w'_{X}$ and $w'_{Y}$, we will determine whether there is a vertex cover $W$ with size exactly $|S|$ such that $Z \subseteq W$.

    Since $W$ contains $Z$ if it exists, we delete all vertices in $Z$ from $G$, and then check whether there exists a vertex cover $W'$ with size $t=|S|-|Z|$ in the remaining graph $G'=G[V(G)\setminus Z]$.
    We observe that $S'= S\setminus Z$ and $T'= T\setminus Z$ remain vertex covers of $G'$.

    Firstly, we partition vertices in $V(G')$ into four vertex sets $S'\setminus T'$, $T'\setminus S'$, $S'\cap T'$, and $U' = V(G')\setminus (S'\cup T')$ (as shown in \Cref{fig:VCRdecomposition}). 
    Since $S'$ and $T'$ are vertex covers, $S'\setminus T'$, $T'\setminus S'$, and $U'$ are independent sets, and there is no edge between $S'\setminus T'$ and $U'$, and $T'\setminus S'$ and $U'$.
    Thus, if we remove all vertices in $S'\cap T'\subseteq S\cap T$ with size at most $\mu - 1$, the remaining graph is bipartite. 
    This is a key property in our algorithm, which is referenced later.
    
    Secondly, we guess a vertex set $A = W' \cap (S'\cap T')$.
    There are at most $O(2^{|S'\cap T'|})=O(2^\mu)$ possible types of such guesses.
    If $A$ is not a vertex cover of $G'[S'\cap T']$, then we abort the current choice $A$ and consider the next one.
    If $v\in S'\cap T'$ is in $A$, we delete $v$; otherwise, we suppose that $N(v) \subseteq W'$, and delete $N[v]$.
    Let $G^*$ be the obtained graph.
    Now, it is supposed that $A\cup N(\notinA)$ is in $W'$, where $\notinA=(S'\cap T')\setminus A$. 
    Since $G^*$ has no vertex in $S'\cap T'$, $G^*$ is bipartite.
    It is known that a minimum vertex cover on a bipartite graph can be computed in polynomial time~\cite{DBLP:journals/jacm/HopcroftPV77,konig1931}.
    Let $B$ be a minimum vertex cover in $G^*$.
    If $|A\cup N(\notinA)|+|B| > t$, then we abort the current choice $A$ and consider the next one.
    Otherwise, there exists a vertex cover $W'\supseteq A \cup N(\notinA) \cup B$ with size exactly $t$ of $G'$, and a vertex cover $W=W' \cup Z$ with size exactly $|S|$ of $G$.
    At least one guess $A$ succeeds if and only if we add an edge between $w'_X$ and $w'_Y$ to $\mathcal{E'}$.
    
    At the end of this algorithm, by the breadth search algorithm, we check whether there is a path between two clique-nodes $w_{X'}$ and $w_{Y'}$ for arbitrarily chosen vertex sets $X'\subseteq S$ and $Y'\subseteq T$ of $G$ with size exactly $\mu$.
    The correctness follows from \Cref{claim:clique_nodes}.

    We estimate the running time of this algorithm. 
    Let $n=|V(G)|$ and $m=|E(G)|$. The constructed clique-compressed reconfiguration graph has exactly $\binom{n}{\mu}=O(n^\mu)$ clique-nodes. For each pair of two distinct clique-nodes, we decide whether there is an edge between them.
    Firstly, we partition the vertices in $V(G')$ into four vertex sets. 
    This can be done in linear time. 
    Then, for at most $O(2^\mu)$ possible guesses regarding a vertex set $S'\cap T'$, we delete vertices of $G^*$ in linear time and run the algorithm that finds a minimum vertex cover of $G^*$ in time $O(\sqrt{n}(n+m))$~\cite{DBLP:journals/jacm/HopcroftPV77,konig1931}. Therefore, we can construct the edge set with size at most $O(n^{2\mu})$ in time $O(n^{2\mu}(n+2^\mu (n+\sqrt{n}(n+m))))=O(2^\mu n^{2\mu+1/2}(n+m))$.
    At last, we can determine whether $(G,S,T,\kTJ)$ is a yes-instance or not by the breadth-first search algorithm in time $O(|\mathcal{V'}|+|\mathcal{E'}|)=O(n^{2\mu})$.
    The total running time of this algorithm is $O(2^\mu n^{2\mu+1/2}(n+m))$.
    Therefore, \prb{VCR} under $\kTJ$ is in {\XP} for general graphs when parameterized by $\mu=|S|-k$.
\end{proof}
\begin{figure}[h]
    \centering
\begin{tikzpicture}[scale=1]

\node[fill=white!10, draw=black, circle, minimum size=3mm] (st0) at (0,0){};
\node[fill=white!10, draw=black, circle, minimum size=3mm] (st1) at (0,0.5){};
\node[fill=white!10, draw=black, circle, minimum size=3mm] (st2) at (0,-0.5){};
\draw[rounded corners] (-0.3,0.8) rectangle(0.3,-0.8);
\node[] at (0,1) {$S'\cap T'$};

\node[fill=white!10, draw=black, circle, minimum size=3mm] (s0) at (1,-1.5){};
\node[fill=white!10, draw=black, circle, minimum size=3mm] (s1) at (2,-1.5){};
\node[fill=white!10, draw=black, circle, minimum size=3mm] (s2) at (3,-1.5){};
\draw (4,-1.5) node{$\cdots$};
\node[fill=white!10, draw=black, circle, minimum size=3mm] (s3) at (5,-1.5){};
\draw[rounded corners] (0.7,-1.8) rectangle(5.3,-1.2);
\node[] at (6,-1.5) {$S'\setminus T'$};

\node[fill=white!10, draw=black, circle, minimum size=3mm] (t0) at (1,1.5){};
\node[fill=white!10, draw=black, circle, minimum size=3mm] (t1) at (2,1.5){};
\node[fill=white!10, draw=black, circle, minimum size=3mm] (t2) at (3,1.5){};
\draw (4,1.5) node{$\cdots$};
\node[fill=white!10, draw=black, circle, minimum size=3mm] (t3) at (5,1.5){};
\draw[rounded corners] (0.7,1.8) rectangle(5.3,1.2);
\node[] at (6,1.5) {$T'\setminus S'$};

\node[fill=white!10, draw=black, circle, minimum size=3mm] (o0) at (-1.5,1.4){};
\node[fill=white!10, draw=black, circle, minimum size=3mm] (o1) at (-1.5,0.7){};
\node[fill=white!10, draw=black, circle, minimum size=3mm] (o2) at (-1.5,0){};
\draw (-1.5,-0.6) node{$\vdots$};
\node[fill=white!10, draw=black, circle, minimum size=3mm] (o3) at (-1.5,-1.4){};
\draw[rounded corners] (-1.8,1.7) rectangle(-1.2,-1.7);
\node[] at (-2,0) {$U'$};

\draw[very thick] (o0)--(st0);
\draw[very thick] (o1)--(st0);
\draw[very thick] (o0)--(st2);
\draw[very thick] (o2)--(st1);
\draw[very thick] (o3)--(st2);

\draw[very thick] (s0)--(st0);
\draw[very thick] (s1)--(st0);
\draw[very thick] (s0)--(st2);
\draw[very thick] (s2)--(st1);
\draw[very thick] (s3)--(st2);

\draw[very thick] (t0)--(st0);
\draw[very thick] (t1)--(st0);
\draw[very thick] (t0)--(st2);
\draw[very thick] (t2)--(st1);
\draw[very thick] (t3)--(st2);

\draw[very thick] (t3)--(s3);
\draw[very thick] (t1)--(s3);
\draw[very thick] (t0)--(s2);
\draw[very thick] (t2)--(s1);
\draw[very thick] (t3)--(s2);

\draw[very thick] (st0)--(st1);

\end{tikzpicture}
    \caption{The partition of the vertices in the graph $G'$. The vertex sets $S'\setminus T'$, $T'\setminus S'$, and $U'=V(G')\setminus (S'\cup T')$ are independent sets, and there is no edge between $U'$ and $S'\bigtriangleup T'$.}
    \label{fig:VCRdecomposition}
\end{figure}
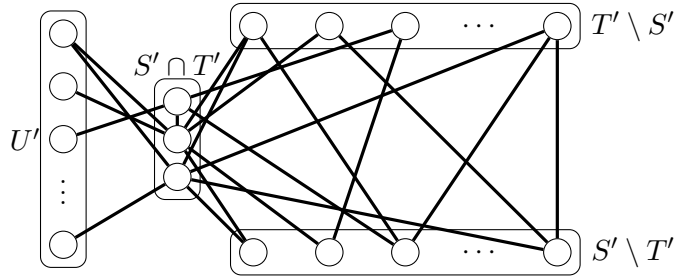

\section{PSPACE-completeness When $k$ is Constant}\label{sec:k_const}
In this section, we investigate the $\PSPACE$-completeness of \prb{ISR} under $\kTJ$ and $\kTS$ when $k$ is fixed.
Note that the computational complexity of \prb{ISR} and \prb{VCR} under $\kTJ$ and $\kTS$ is the same when $k$ is fixed, due to their complementary relationship.

\subsection{Planar Graphs}
This subsection is dedicated to the following results.
\begin{theorem}\label{PSPACE-comp-planar}
    Let $k\geq 2$ be any fixed positive integer. \prb{ISR} under $\Rule\in\{\kTS,\kTJ\}$ is $\PSPACE$-complete for planar graphs of maximum degree $3$ and bounded bandwidth.
\end{theorem}
To prove \Cref{PSPACE-comp-planar}, we construct a reduction from \prb{Nondeterministic Constraint Logic}, which was introduced by Hearn and Demaine~\cite{HD05} and has been used to prove the $\PSPACE$-hardness of reconfiguration problems, including \prb{ISR} under $\TS$~\cite{HD05}.

In this section, we first present the definition of Nondeterministic Constraint Logic in \Cref{subsubsec:DefOfNCL}. We then introduce the gadgets used in our reduction in \Cref{subsubsec:NCLgadget}. Finally, we describe our polynomial-time reduction in \Cref{subsubsec:NCLtoISR} and prove its correctness in \Cref{subsubsec:CorrectnessOfNCLtoISR}.

\subsubsection{Definition of Nondeterministic Constraint Logic}\label{subsubsec:DefOfNCL}

An NCL ``machine'' is defined by an undirected graph together with an assignment of edge weights drawn from $\{1,2\}$.
An \emph{(NCL) configuration} of this machine is an orientation of the edges such that, at every vertex, the total weight of the incoming arcs is at least two.
\Cref{fig:NCL}(a) illustrates an example configuration of an NCL machine, where weight-2 edges are shown as thick (blue) lines and weight-1 edges as thin (red) lines.
Two NCL configurations are said to be \emph{adjacent} if they differ only in the orientation of a single edge.
In the {\NCL} problem, we are given an NCL machine together with two of its configurations.
It is known that determining whether there exists a sequence of adjacent configurations that transforms one into the other is $\PSPACE$-complete~\cite{HD05}.

An NCL machine is called an \textsc{and/or} \emph{constraint graph} if it consists solely of two types of vertices, namely \emph{NCL \textsc{and} vertices} and \emph{NCL \textsc{or} vertices}.
A vertex of degree $3$ is defined as an \emph{NCL \textsc{and} vertex} if its three incident edges have weights $1$, $1$, and $2$ (see \Cref{fig:NCL}(b)).
An NCL \textsc{and} vertex $u$ behaves analogously to a logical \textsc{and}: the weight-2 edge can be directed outward from $u$ if and only if both weight-1 edges are directed inward to $u$.
Note, however, that the weight-2 edge is not necessarily directed outward even when both weight-1 edges are directed inward.
Similarly, a vertex of degree $3$ is defined as an \emph{NCL \textsc{or} vertex} if all three incident edges have weight $2$ (see \Cref{fig:NCL}(c)).
An NCL \textsc{or} vertex $v$ behaves like a logical \textsc{or}: one of its three edges can be directed outward from $v$ if and only if at least one of the other two edges is directed inward.

For example, the NCL machine in \Cref{fig:NCL}(a) is an \textsc{and/or} constraint graph. From now on, we call an \textsc{and/or} constraint graph simply an \emph{NCL machine}, and call an edge in an NCL machine an \emph{NCL edge}.
\prb{Nondeterministic Constraint Logic} remains $\PSPACE$-complete even if an input NCL machine is a planar graph of bounded bandwidth and maximum degree 3~\cite{W18}.

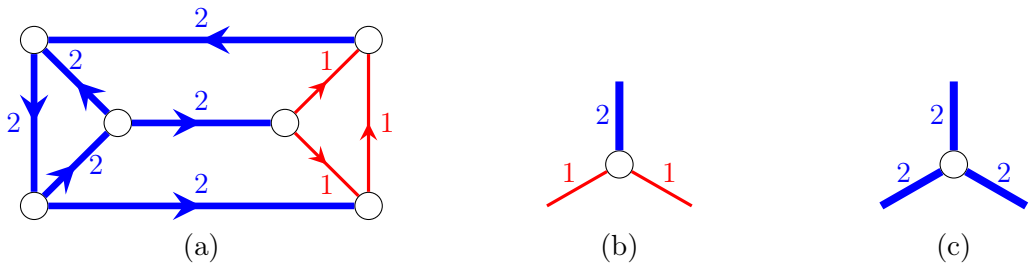
\begin{figure}[t]
	\centering
\begin{tikzpicture}[scale=1.1]
	\node[fill=white!10, draw=black, circle, minimum size=3.5mm] (N1) at (-9,1.5){};
	\node[fill=white!10, draw=black, circle, minimum size=3.5mm] (N2) at (-9,-0.5){};
	\node[fill=white!10, draw=black, circle, minimum size=3.5mm] (N3) at (-8,0.5){};
	\node[fill=white!10, draw=black, circle, minimum size=3.5mm] (N4) at (-5,1.5){};
	\node[fill=white!10, draw=black, circle, minimum size=3.5mm] (N5) at (-5,-0.5){};
	\node[fill=white!10, draw=black, circle, minimum size=3.5mm] (N6) at (-6,0.5){};
	
	\draw[line width=2.5pt, blue, >=stealth, postaction={decorate}, decoration={markings,mark=at position 0.5 with {\arrow[scale=1.3]{>}}}] (N1)  -- (N2) node[midway, left] {2};
	\draw[line width=2.5pt, blue, >=stealth, postaction={decorate}, decoration={markings,mark=at position 0.5 with {\arrow[scale=1.3]{>}}}] (N3)  -- (N1) node[midway, above] {2};
	\draw[line width=2.5pt, blue, >=stealth, postaction={decorate}, decoration={markings,mark=at position 0.5 with {\arrow[scale=1.3]{>}}}] (N4)  -- (N1) node[midway, above] {2};
	\draw[line width=2.5pt, blue, >=stealth, postaction={decorate}, decoration={markings,mark=at position 0.5 with {\arrow[scale=1.3]{>}}}] (N2)  -- (N3) node[midway, right] {2};
	\draw[line width=2.5pt, blue, >=stealth, postaction={decorate}, decoration={markings,mark=at position 0.5 with {\arrow[scale=1.3]{>}}}] (N2)  -- (N5) node[midway, above] {2};
	\draw[line width=2.5pt, blue, >=stealth, postaction={decorate}, decoration={markings,mark=at position 0.5 with {\arrow[scale=1.3]{>}}}] (N3)  -- (N6) node[midway, above] {2};
	
	\draw[very thick, red, >=stealth, postaction={decorate}, decoration={markings,mark=at position 0.5 with {\arrow[scale=1.3]{>}}}] (N5) -- (N4) node[midway, right] {1};
	\draw[very thick, red, >=stealth, postaction={decorate}, decoration={markings,mark=at position 0.5 with {\arrow[scale=1.3]{>}}}] (N6) -- (N4)  node[midway, above] {1};
	\draw[very thick, red, >=stealth, postaction={decorate}, decoration={markings,mark=at position 0.5 with {\arrow[scale=1.3]{>}}}] (N6) -- (N5)  node[midway, below] {1};
	
	\node[] at (-7,-1) {(a)};
	
	\node[fill=white!10, draw=black, circle, minimum size=3.5mm] (or) at (-2,0){};
	
	\draw[line width=3pt, blue] (or) node at (-2.2,0.6) {2} --(-2, 1);
	\draw[very thick, red] (or) node at (-1.4,-0.1) {1} --(-2+0.866, -0.5);
	\draw[very thick, red] (or) node at (-2.6,-0.1) {1} --(-2-0.866, -0.5);
	
	\node[] at (-2,-1) {(b)};

	\node[fill=white!10, draw=black, circle, minimum size=3.5mm] (and) at (2,0){};
	
	\draw[line width=3pt, blue] (and) node at (1.8,0.6) {2} --(2, 1);
	\draw[line width=3pt, blue] (and) node at (1.4,-0.1) {2} --(2+0.866, -0.5);
	\draw[line width=3pt, blue] (and) node at (2.6,-0.1) {2} --(2-0.866, -0.5);
	
	\node[] at (2,-1) {(c)};
	
\end{tikzpicture}
	\caption{(a) A configuration of an \prb{NCL} machine, (b) an \prb{NCL} \textsc{And} vertex, and (c) an \prb{NCL} \textsc{Or} vertex. 
    Edges of weight $2$ are shown as thick (blue) lines, and edges of weight $1$ as thin (red) lines.
    The NCL machine in (a) is an \textsc{and/or} constraint graph.}
	\label{fig:NCL}
\end{figure}

\subsubsection{Gadgets}\label{subsubsec:NCLgadget}

Suppose that we are given an instance of {\NCL}, that is, an NCL machine and two configurations of the machine.
We will replace each of NCL edges and NCL \textsc{and/or} vertices with its corresponding gadget.
For an NCL edge $e = uv$, the corresponding gadget for $e$ has two specific vertices, namely a \emph{$(u,e)$-connector} and a \emph{$(v,e)$-connector}, as illustrated in \Cref{fig:NCL_connect}.
We sometimes simply refer to them as \emph{connectors}.
The gadget for $e$ shares the $(u,e)$-connector and the $(v,e)$-connector with the corresponding gadgets for $u$ and $v$, respectively.
Thus, each gadget corresponding to a vertex has three connectors.
Note that each pair of gadgets shares at most one connector, and their edges are disjoint.

Our reduction involves constructing three types of gadgets which correspond to NCL edges and NCL \textsc{and/or} vertices, as shown in \Cref{fig:NCLtoISR}.
Below, we explain the construction of each gadget.

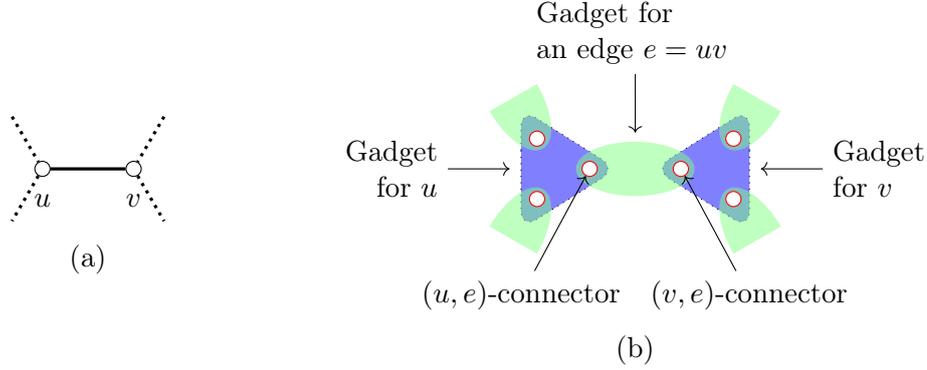
\begin{figure}[t]
	\centering
	\begin{tikzpicture}[scale=1.2]

\node[fill=white!10, draw=black, circle, minimum size=2mm, inner sep=0pt] (u) at (-4,0){};
\node[] at (-4,-0.35) {$u$};
\node[fill=white!10, draw=black, circle, minimum size=2mm, inner sep=0pt] (v) at (-3,0){};
\node[] at (-3,-0.35) {$v$};

\draw[very thick] (u)--(v);
\draw[very thick, dotted] (-4-1/3,1.73/3)--(u)--(-4-1/3,-1.73/3);
\draw[very thick, dotted] (-3+1/3,1.73/3)--(v)--(-3+1/3,-1.73/3);

\node[] at (-3.5, -1) {(a)};


\filldraw[fill=blue!50, dotted, rounded corners=7pt](2-1.73/2.3,0)--(2-1.73/2.3,1/1.5)--(2+0.3,0)--(2-1.73/2.3,-1/1.5)-- cycle;
\filldraw[fill=blue!50, dotted, rounded corners=7pt](3+1.73/2.3,0)--(3+1.73/2.3,1/1.5)--(3-0.3,0)--(3+1.73/2.3,-1/1.5)-- cycle;
\fill[fill=green!50, opacity=0.5] (2.5,0) ellipse [x radius=0.65, y radius=0.3];
\begin{scope}[rotate=-60]
\fill[green!50,opacity=0.5]
  (-0.1,1.1) arc[start angle=-90, end angle=90, x radius=0.65, y radius=0.3] -- cycle;
\end{scope}
\begin{scope}[rotate=60]
\fill[green!50,opacity=0.5]
  (-0.1,-1.7) arc[start angle=-90, end angle=90, x radius=0.65, y radius=0.3] -- cycle;
\end{scope}

\begin{scope}[rotate=60]
\fill[green!50,opacity=0.5]
  (2.6,-2.63) arc[start angle=90, end angle=270, x radius=0.65, y radius=0.3] -- cycle;
\end{scope}
\begin{scope}[rotate=-60]
\fill[green!50,opacity=0.5]
  (2.6,3.23) arc[start angle=90, end angle=270, x radius=0.65, y radius=0.3] -- cycle;
\end{scope}

\node[fill=white!10, draw=red, circle, minimum size=2mm, inner sep=0pt] (u_1) at (2,0){};
\node[fill=white!10, draw=red, circle, minimum size=2mm, inner sep=0pt] (u_2) at (2-1.73/3,1/3){};
\node[fill=white!10, draw=red, circle, minimum size=2mm, inner sep=0pt] (u_3) at (2-1.73/3,-1/3){};

\node[fill=white!10, draw=red, circle, minimum size=2mm, inner sep=0pt] (v_1) at (3,0){};
\node[fill=white!10, draw=red, circle, minimum size=2mm, inner sep=0pt] (v_2) at (3+1.73/3,1/3){};
\node[fill=white!10, draw=red, circle, minimum size=2mm, inner sep=0pt] (v_3) at (3+1.73/3,-1/3){};

\node[text width = 2cm, align=right] (text_u) at (-0.5,0) {Gadget\\for $u$};
\draw[->] (text_u) -- (2-1.73/3-0.3,0); 

\node[] (text_ue_connector) at (1.25,-1.4) {$(u,e)$-connector};
\draw[->] (text_ue_connector) -- (u_1); 

\node[text width = 2cm, align=left] (text_v) at (5.5,0) {Gadget\\for $v$};
\draw[->] (text_v) -- (3+1.73/3+0.3,0); 

\node[] (text_ve_connector) at (3.75,-1.4) {$(v,e)$-connector};
\draw[->] (text_ve_connector) -- (v_1); 

\node[] at (2.5, -2) {(b)};

\node[text width = 2.6cm, align=left] (text_port) at (2.5,1.5) {Gadget for\\ an edge $e=uv$};
\draw[->] (text_port) -- (2.5, 0.4);

\end{tikzpicture}
	\caption{(a) An NCL edge $uv$, and (b) its corresponding gadgets, where the connectors are depicted by (red) circles.}
	\label{fig:NCL_connect}
\end{figure}

\begin{figure}[t]
	\centering
\begin{tikzpicture}[scale=1]

\filldraw[fill=green!50, draw= green, opacity=0.5] (-3.5,0) ellipse [x radius=1.4, y radius=0.8];

\node[fill=white!10, draw=red, circle, minimum size=4mm, inner sep=0pt] (e1) at (-4.5,0){};
\node[fill=black!100, draw=black, circle, minimum size=2.5mm, inner sep=0pt] (te1) at (-4.5,0){};

\node[fill=white!10, draw=black, circle, minimum size=4mm, inner sep=0pt] (e2) at (-3.9,0.4){};

\node[fill=white!10, draw=black, circle, minimum size=4mm, inner sep=0pt] (e3) at (-3.1,0.4){};
\node[fill=black!100, draw=black, circle, minimum size=2.5mm, inner sep=0pt] (te1) at (-3.1,0.4){};

\node[fill=white!10, draw=red, circle, minimum size=4mm, inner sep=0pt] (e4) at (-2.5,0){};

\draw[very thick] (e1) to [out=-45, in=225] (e4);
\draw[very thick] (e1) -- (e2);
\draw[ultra thick, draw=red] (e2) -- (e3);
\draw[very thick] (e3) -- (e4);

\node[] (text_u) at (-4.5,-1.25) {$(u,e)$-connector};
\draw[->] (text_u) -- (e1); 

\node[] (text_u) at (-2.5,1.25) {$(v,e)$-connector};
\draw[->] (text_u) -- (e4);

\node[] at (-3.5,-2.5) {(a)};

\filldraw[dashed, fill =blue!30, thick, rounded corners=15pt] (0,-1)--(1.5,1.5)--(3,-1)--cycle;

\node[fill=white!10, draw=red, circle, minimum size=4mm, inner sep=0pt] (v2) at (1.5,0.75){};

\node[fill=white!10, draw=black, circle, minimum size=4mm, inner sep=0pt] (v3) at (1.5,0.2){};

\node[fill=white!10, draw=black, circle, minimum size=4mm, inner sep=0pt] (v4) at (1.5,-0.35){};
\node[fill=black!100, draw=black, circle, minimum size=2.5mm, inner sep=0pt] (tv4) at (1.5,-0.35){};

\node[fill=white!10, draw=red, circle, minimum size=4mm, inner sep=0pt] (v5) at (0.6,-0.65){};

\node[fill=white!10, draw=red, circle, minimum size=4mm, inner sep=0pt] (v6) at (2.4,-0.65){};

\draw[very thick] (v2)--(v3);
\draw[very thick] (v3)--(v4);
\draw[very thick] (v4)--(v5);
\draw[very thick] (v4)--(v6);

\node[] (text_u) at (1.75,1.75) {$(u,e)$-connector};
\draw[->] (text_u) -- (v2); 

\node[] (text_u) at (0.5,-1.5) {$(u,e_1)$-connector};
\draw[->] (text_u) -- (v5); 

\node[] (text_u) at (2.5,-2) {$(u,e_2)$-connector};
\draw[->] (text_u) -- (v6); 

\node[] at (1.5,-2.5) {(b)};

\filldraw[dashed, fill =blue!30, thick, rounded corners=15pt] (4.5,-1)--(6,1.5)--(7.5,-1)--cycle;

\node[fill=white!10, draw=red, circle, minimum size=4mm, inner sep=0pt] (u2) at (6,0.75){};

\node[fill=white!10, draw=black, circle, minimum size=4mm, inner sep=0pt] (u3) at (6,0.2){};
\node[fill=black!100, draw=black, circle, minimum size=2.5mm, inner sep=0pt] (tu3) at (6,0.2){};

\node[fill=white!10, draw=black, circle, minimum size=4mm, inner sep=0pt] (u4) at (5.6,-0.4){};

\node[fill=white!10, draw=black, circle, minimum size=4mm, inner sep=0pt] (u5) at (6.4,-0.4){};

\node[fill=white!10, draw=red, circle, minimum size=4mm, inner sep=0pt] (u6) at (5.1,-0.65){};

\node[fill=white!10, draw=red, circle, minimum size=4mm, inner sep=0pt] (u7) at (6.9,-0.65){};

\draw[very thick] (u2)--(u3);
\draw[very thick] (u3)--(u4);
\draw[very thick] (u3)--(u5);
\draw[very thick] (u4)--(u5);
\draw[very thick] (u4)--(u6);
\draw[very thick] (u5)--(u7);

\node[] (text_u) at (1.75+4.5,1.75) {$(v,e)$-connector};
\draw[->] (text_u) -- (u2); 

\node[] (text_u) at (0.5+4.5,-1.5) {$(v,e_a)$-connector};
\draw[->] (text_u) -- (u6); 

\node[] (text_u) at (2.5+4.5,-2) {$(v,e_b)$-connector};
\draw[->] (text_u) -- (u7); 

\node[] at (6,-2.5) {(c)};

\node[text width = 1.25cm, align=center] (text_internaltoken) at (3.75,0.75) {internal\\ token};
\draw[->] (text_internaltoken) -- (tv4); 
\draw[->] (text_internaltoken) -- (tu3); 
\end{tikzpicture}
	\caption{(a) An \prb{NCL} edge gadget when $k=2$. For $k\geq 3$, the red edge is replaced by a path with $2(k-2)$ vertices, and additional $k-2$ tokens are placed. (b) An \prb{NCL} \gad{and} gadget and (c) an \prb{NCL} \gad{or} gadget. The connectors are depicted as red circles.}
	\label{fig:NCLtoISR}
\end{figure}
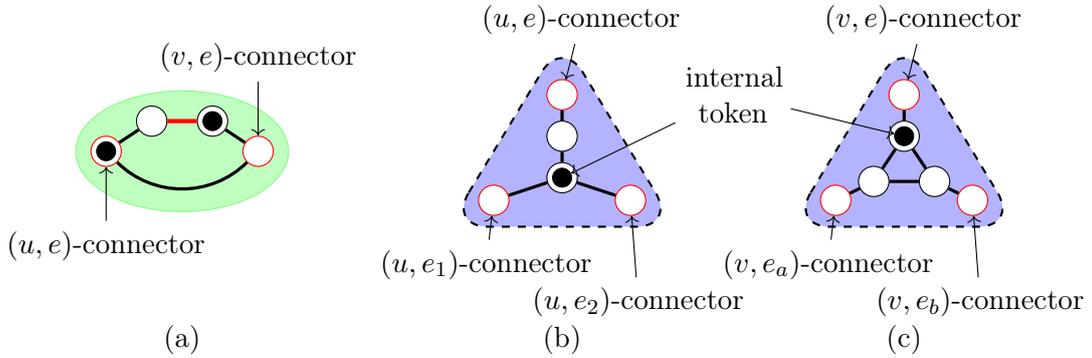

\vspace{5pt}
\noindent \textbf{NCL edge Gadget.}
Our NCL edge gadget consists of a cycle with $2k$ vertices and $k$ tokens.  
The \prb{NCL} edge gadget has two connectors: the $(u,e)$-connector and the $(v,e)$-connector, where $e=uv$ is a corresponding \prb{NCL} edge.
Note that the $(u,e)$-connector and the $(v,e)$-connector are consecutive in the cycle.
\Cref{fig:NCLtoISR}(a) illustrates the \prb{NCL} edge gadget for the case $k = 2$.  
For $k \geq 3$, the red edge is replaced by a path with $2(k - 2)$ vertices, and an additional $k - 2$ tokens are placed accordingly.

 \vspace{5pt}
\noindent \textbf{NCL} \textsc{and} \textbf{Gadget.}
\Cref{fig:NCLtoISR}(b) illustrates our \prb{NCL} \gad{and} gadget for an \prb{NCL} \gad{and} vertex \( u \).  
The gadget has three connectors: the \( (u, e) \)-connector, the \( (u, e_1) \)-connector, and the \( (u, e_2) \)-connector, where \( e_1 \) and \( e_2 \) are weight-1 \prb{NCL} edges, and \( e \) is a weight-2 \prb{NCL} edge.
Furthermore, the gadget has two vertices that are adjacent to each other and distinct from the connectors, referred to as \emph{internal vertices}: one is connected to the $(u,e)$-connector, and the other is connected to both the \( (u, e_1) \)- and \( (u, e_2) \)-connectors.
The \prb{NCL} \gad{and} gadget contains a single token, called an \emph{internal token}, excluding the tokens provided by the \prb{NCL} edge gadgets.

\vspace{5pt}
\noindent \textbf{NCL} \textsc{or} \textbf{Gadget.}
\Cref{fig:NCLtoISR}(c) illustrates our \prb{NCL} \gad{or} gadget for an \prb{NCL} \gad{or} vertex \( v \).  
The gadget has three connectors: the \( (v, e) \)-connector, the \( (v, e_a) \)-connector, and the \( (v, e_b) \)-connector, where all edges \( e \), \( e_a \), and \( e_b \) are weight-2 \prb{NCL} edges.
Furthermore, the gadget has three mutually adjacent vertices that are distinct from the connectors, referred to as \emph{internal vertices}, each connected to a different connector.
The \prb{NCL} \gad{or} gadget also contains an internal token, excluding the tokens provided by the \prb{NCL} edge gadgets.

\subsubsection{Reduction}\label{subsubsec:NCLtoISR}
Let $(M,C_s,C_t)$ be an instance of {\NCL}, where $M$ is an NCL machine, and $C_s$ and $C_t$ are two configurations of $M$.
We will construct an instance $(G,I,J,\Rule\in\{\kTJ,\kTS\})$ of \prb{ISR}.
As illustrated in \Cref{fig:NCL_connect}, we replace each \prb{NCL} edge and each \prb{NCL} \textsc{and/or} vertex with its corresponding gadget.  
The graph \( G \) is constructed from the disjoint collection of gadgets by identifying the \( (u, e) \)-connector in the edge gadget for \( e \) with the \( (u, e) \)-connector in the \prb{NCL} \textsc{and/or} gadget for \( u \).

Note that each gadget has maximum degree \( 3 \).  
In particular, the connectors in the \prb{NCL} edge gadgets have degree \( 2 \), while those in the \prb{NCL} \gad{and/or} gadgets have degree \( 1 \).
Thus, the constructed graph $G$ is a graph of maximum degree $3$.
Moreover, since both the original NCL machine \( M \) and all the gadgets are planar, the resulting graph \( G \) is also planar.  
Similarly, because the bandwidth of \( M \) is bounded by some constant and each gadget has a bounded size, the bandwidth of \( G \) is also bounded by some constant.
Therefore, $G$ is a planar graph of maximum degree $3$ and the bandwidth of $G$ is bounded.

Next, we construct two independent sets \( I \) and \( J \) of \( G \), corresponding to two given \prb{NCL} configurations \( C_s \) and \( C_t \) of the \prb{NCL} machine, respectively.  
In each of \( I \) and \( J \), every \prb{NCL} edge gadget contains exactly \( k \) tokens.  
For each \prb{NCL} edge \( e = uv \), if it is directed toward \( u \), we place the tokens starting from the \((e, v)\)-connector and continuing clockwise, skipping every other position; otherwise, we start from the \((e, u)\)-connector and proceed in the same manner.

For each \prb{NCL} \gad{and/or} gadget, we place an internal token on an internal vertex.
We claim that this token can be placed for the following reasons.
Consider an \prb{NCL} \gad{and} gadget corresponding to an \prb{NCL} \gad{and} vertex $u$, which has one incident weight-2 edge $e$ and two incident weight-1 edges $e_1$ and $e_2$. 
For the sake of contradiction, suppose that we are unable to place the internal token. 
This implies that the $(u,e)$-connector and either the $(u,e_1)$-connector or the $(u,e_2)$-connector are already occupied by tokens placed according to the assignment of tokens to the \prb{NCL} edge gadgets.
However, this situation would mean that the weight-2 edge $e$ and at least one of the weight-1 edges $e_1$ or $e_2$ are directed outward from $u$, contradicting the assumption that $C_s$ (and $C_t$) is a valid configuration of the \prb{NCL} machine $M$.

Similarly, consider an \prb{NCL} \gad{or} gadget corresponding to an \prb{NCL} \gad{or} vertex $v$ with three incident weight-2 edges $e$, $e_a$, and $e_b$. If we cannot place the token, then the $(v,e)$-, $(v,e_a)$-, and $(v,e_b)$-connectors must be occupied by tokens.
This would imply that all three weight-2 edges are directed outward from $v$, contradicting the assumption that $C_s$ (and $C_t$) is a valid configuration of the \prb{NCL} machine $M$.

Therefore, we can place an additional token on an internal vertex in each \prb{NCL} \gad{and/or} gadget. 
Although there may be multiple options, we arbitrarily choose one of them.

This completes the construction of our corresponding instance of \prb{ISR} under $\kTJ$ and $\kTS$. 
The construction can be done in polynomial time.

\subsubsection{Correctness}\label{subsubsec:CorrectnessOfNCLtoISR}
Before proving the correctness of our reduction, we first provide an overview.
Specifically, we explain the correspondence between the orientations of an \prb{NCL} machine and the independent sets of the corresponding graph, as well as the behavior of each gadget.

We interpret the direction of an NCL edge \( e = uv \) as \emph{inward} to \( u \) if the \((u,e)\)-connector is \emph{not} occupied by a token. 
Conversely, we interpret the direction of \( e = uv \) as \emph{outward} from \( u \) if the \((u,e)\)-connector \emph{is} occupied by a token.

Since the \((u,e)\)-connector and \((v,e)\)-connector are adjacent, they cannot be occupied by tokens simultaneously. 
Furthermore, each \prb{NCL} edge gadget contains exactly \( k \) tokens, and each \prb{NCL} \gad{and/or} gadget contains exactly one internal token placed on an internal vertex. 
Thus, there are $n+km$ tokens in total, which is equal to the size of the maximum independent set of $G$, where $n$ denotes the number of vertices and $m$ denotes the number of edges of $M$.
Since each gadget has no space to accommodate additional tokens, every configuration of $n+km$ tokens must preserve this structure: Each \prb{NCL} edge gadget contains \( k \) tokens, and each \prb{NCL} \gad{and/or} gadget contains one internal token on an internal vertex.
Therefore, for every \prb{NCL} edge gadget for \( e = uv \), exactly one of the two connectors \((u,e)\)-connector or \((v,e)\)-connector is occupied by a token at any time.
For this reason, we can construct a configuration of $M$ from independent sets of $G$ with size $n+km$.

Moreover, we have to move $k$ tokens simultaneously to move the tokens in each \prb{NCL} edge gadget.
This corresponds to the reversal of the direction of the corresponding \prb{NCL} edge.

Next, we explain briefly the behavior of each \prb{NCL} \gad{and/or} gadget.

For an \prb{NCL} \gad{and} vertex \( u \), the weight-2 edge $e$ can be directed outward from \( u \) if and only if both weight-1 edges $e_1$ and $e_2$ are directed inward toward \( u \).
Our \prb{NCL} \gad{and} gadget faithfully simulates this behavior as follows:  
The \((u, e)\)-connector can be occupied by a token if and only if the internal token can occupy the degree-3 internal vertex.  
In this case, no token occupies either the \((u, e_1)\)-connector or the \((u, e_2)\)-connector.  
This configuration precisely corresponds to the case where the edge gadgets for \( e_1 \) and \( e_2 \) are directed inward toward \( u \).

For an \prb{NCL} \gad{or} vertex $v$, the only forbidden orientation is when all three incident \prb{NCL} edges \( e \), \( e_a \), and \( e_b \) are simultaneously directed outward from \( v \).  
This corresponds to the case where tokens from the edge gadgets are placed on all three connectors.
Our \prb{NCL} \gad{or} gadget prevents such a configuration, since otherwise the internal token could not be placed while maintaining an independent set.

Lastly, to complete our polynomial-time reduction, we give the following \Cref{lem:NCL_ISR_correctness}.

\begin{lemma}\label{lem:NCL_ISR_correctness}
    $(M,C_s,C_t)$ is a yes-instance of {\NCL} if and only if $(G,I,J,\Rule\in\{\kTJ,\kTS\})$ is a yes-instance of \prb{ISR}.
\end{lemma}
\begin{proof}
    We first focus on $\kTS$.
    Later, we will show that in the constructed graph, any token movement under $\kTJ$ can be interpreted as one under $\kTS$. Hence, our reduction also applies to $\kTJ$.
    
    We first prove the only if direction.
    Suppose that $(M, C_s, C_t)$ is a yes-instance of \prb{NCL}.
    Then, there exists a sequence of \prb{NCL} configurations from $C_s$ to $C_t$.

    Consider any two consecutive configurations \( C_{i-1} \) and \( C_i \) in this sequence.  
    Since \( C_{i-1} \) and \( C_i \) are adjacent, exactly one edge \( uv \) changes its direction between \( C_{i-1} \) and \( C_i \).  
    Let \( I_{i-1} \) and \( I_i \) be independent sets of \( G \) constructed from \( C_{i-1} \) and \( C_i \), respectively, in the same manner as \( I \) and \( J \). Recall that every \prb{NCL} edge gadget contains exactly \( k \) tokens.  
    For each \prb{NCL} edge \( e = uv \), if it is directed toward \( u \), we place the tokens starting from the \((e, v)\)-connector and continuing clockwise, skipping every other position; otherwise, we start from the \((e, u)\)-connector and proceed in the same manner.
    For each \prb{NCL} \gad{and/or} gadget, we place an internal token on an internal vertex.

    By construction, all \prb{NCL} edge gadgets, except the one corresponding to \( uv \), have identical token configurations in both \( I_{i-1} \) and \( I_i \), since only the direction of \( uv \) changes.  
    Thus, we can obtain an independent set \( I'_{i-1} \) from \( I_{i-1} \) such that \( I'_{i-1} \symdif I_i \) contains exactly the vertices of the \prb{NCL} edge gadget for \( uv \), by moving an internal token in each \prb{NCL} \gad{and/or} gadget.  
    This movement of an internal token does not affect any other gadgets.

    Since the edge gadget forms a cycle and $|I'_{i-1} \symdif I_i| \leq 2k$, it follows that $I'_{i-1}$ and $I_i$ are adjacent under $\kTS$ by moving the tokens in $I'_{i-1}$ on the \prb{NCL} edge gadget to their adjacent positions in the clockwise direction.
    Hence, $I_{i-1}$ and $I_i$ are reconfigurable.
    This shows that any single edge reversal in an \prb{NCL} configuration can be simulated by a reconfiguration sequence of independent sets.
    By repeating this step for the entire sequence from $C_s$ to $C_t$, we obtain a reconfiguration sequence from $I$ to $J$.
    Therefore, the instance $(G,I,J,\kTS)$ is a yes-instance of \prb{ISR}.

    Conversely, suppose that \( (G, I, J, \kTS) \) is a yes-instance of \prb{ISR}.  
    Then, there exists a reconfiguration sequence \( \sigma = \langle I = I_0, I_1, \ldots, I_\ell = J \rangle \) between \( I \) and \( J \) under the \(\kTS\) rule.  
    As discussed earlier, in every independent set in the sequence, each \prb{NCL} edge gadget contains exactly \( k \) tokens, and exactly one of its connectors is occupied by a token.  
    By our construction, each independent set \( I_i \) in \( \sigma \) corresponds to a valid configuration \( C_i \) of the \prb{NCL} machine.  
    Consider the sequence \( C_0, C_1, \ldots, C_\ell \) where each \( C_i \) corresponds to the independent set \( I_i \).  
    Since any two consecutive independent sets \( I_{i-1} \) and \( I_i \) are adjacent under the \(\kTS\) rule, and all \prb{NCL} edge gadgets maintain exactly \( k \) tokens, at most one \prb{NCL} edge gadget differs in the token placements corresponding to them.  
    Thus, any two consecutive configurations \( C_{i-1} \) and \( C_i \) are either identical or differ in the direction of exactly one edge \( uv \).  
    By removing consecutive duplicates if necessary, we obtain a desired sequence of configurations of $M$ between \( C_s \) and \( C_t \) such that any two consecutive configurations are adjacent.  
    Therefore, \( (M, C_s, C_t) \) is a yes-instance of \prb{NCL}.

    We conclude the proof by showing that our reduction also applies to $\kTJ$.  
    Let \( I \) and \( I' \) be two independent sets that are adjacent under $\kTJ$.  
    Since each \prb{NCL} edge gadget always contains exactly \( k \) tokens, token movement within an edge gadget under $\kTJ$ coincides with that under $\kTS$.  
    Moreover, in both \( I \) and \( I' \), each \prb{NCL} \gad{and/or} gadget has exactly one token on the internal vertices.  
    Thus, any internal token movement can be regarded as occurring within a single \prb{NCL} \gad{and/or} gadget under $\kTS$.  

    Therefore, our reduction also applies to $\kTJ$.
    Our proof is completed.
\end{proof}

\subsection{Line Graphs and Claw-free Graphs}
In this section, we show that \prb{ISR} under $\Rule\in \{2\text{-}\TJ, 2\text{-}\TS\}$ is $\PSPACE$-complete even for line graphs, which contrasts that \prb{ISR} under $\Rule\in\{\TJ,\TS\}$ can be solved in polynomial time for line graphs~(more generally, claw-free graphs)~\cite{BonsmaKW14}.
For a graph \( G \), its \emph{line graph} \( L \) is defined as follows: each vertex of \( L \) corresponds to an edge of \( G \), and two vertices in \( L \) are adjacent if and only if their corresponding edges in \( G \) share a common endpoint.

\begin{theorem}\label{PSPACE-comp-claw-free}
    \prb{ISR} under $\Rule\in \{2\text{-}\TJ, 2\text{-}\TS\}$ is $\PSPACE$-complete for line graphs.
\end{theorem}

In the proof of \Cref{PSPACE-comp-claw-free}, we reduce \prb{Perfect Matching Reconfiguration}, which is known to be $\PSPACE$-complete on bipartite graphs of maximum degree $5$ and bounded bandwidth~\cite{DBLP:conf/mfcs/BonamyBHIKMMW19}, to our problem.

Let $G=(V, E)$ be a graph. 
An edge set $M\subseteq E$ is called a \emph{matching} of $G$ if any two edges in $M$ do not share their endpoint.
A matching is \emph{perfect} if it covers all vertices of $G$.
For any two perfect matchings $M$ and $M'$ of a graph, the symmetric difference of them consists of even-length disjoint cycles.
Thus, if $|M\bigtriangleup M'|=4$, then $M\bigtriangleup M'$ induces a cycle of length 4.
A \emph{flip} operation is defined as a transformation of $M$ into $M'$ such that $|M\bigtriangleup M'|=4$.
We say that $M$ and $M'$ are \emph{adjacent} under the flip operation.

In the \prb{Perfect Matching Reconfiguration} problem, given two perfect matchings \(M_s\) and \(M_t\) of a graph \(G\), the question is whether there exists a sequence \(M_0, \ldots, M_\ell\) of perfect matchings of \(G\) such that \(M_0 = M_s\), \(M_\ell = M_t\), and each consecutive pair \(M_i, M_{i+1}\) is adjacent under the flip operation.

Now, we show \Cref{PSPACE-comp-claw-free}.
\begin{proof}[Proof of \Cref{PSPACE-comp-claw-free}]
    To prove the $\PSPACE$-hardness, we use a polynomial-time reduction from \prb{Perfect Matching Reconfiguration}.
    Let $(G=(V,E), M_s, M_t)$ be an instance of \prb{Perfect Matching Reconfiguration}.
    Let $L=(E,E')$ be the line graph of $G$, where $E'=\{ef \colon e,f\in E \land e\cap f\neq \emptyset\}$.
    Due to the correspondence between matchings in a graph and independent sets in its line graph, $M_s$ and $M_t$ are independent sets of $L$.
    Then, let $(L,M_s,M_t)$ be an instance of \prb{ISR} under $\Rule\in \{2\text{-}\TJ, 2\text{-}\TS\}$.

    To complete our reduction, we will show that $(G, M_s, M_t)$ is a yes-instance of \prb{Perfect Matching Reconfiguration} if and only if $(L,M_s,M_t)$ is a yes-instance of \prb{ISR} under $\Rule$.

    Firstly, suppose that there is a sequence \(M_0, M_1,\ldots, M_\ell\) of perfect matchings of \(G\) such that \(M_0 = M_s\), \(M_\ell = M_t\), and each consecutive pair \(M_i, M_{i+1}\) is adjacent under the flip operation.
    Then, the sequence is also a sequence of independent sets of $L$. 
    For two consecutive perfect matchings $M$ and $M'$, since $M\bigtriangleup M'$ induces the cycle of length 4 of $G$ and the line graph of a cycle is also a cycle with the same length, two independent sets $M$ and $M'$ of $L$ are adjacent under $\Rule$.
    Thus, \(M_0, M_1,\ldots, M_\ell\) is a reconfiguration sequence of independent sets of $L$ between $M_s$ and $M_t$ under $\Rule$. Therefore, $(L,M_s,M_t)$ is a yes-instance of \prb{ISR} under $\Rule$.

    Conversely, suppose that there is a reconfiguration sequence $\sigma=\langle M_s=M_0, M_1, \ldots, M_\ell=M_t \rangle$ of independent sets of $L$ under $\Rule$. 
    Then, \(M_0, M_1,\ldots, M_\ell\) is also a sequence of matchings of $G$.
    All matchings in $\sigma$ of $G$ have the same size with perfect matchings $M_s$ and $M_t$.
    Thus, all matchings in \(M_0, M_1,\ldots, M_\ell\) of $G$ are perfect.
    Since two consecutive independent sets $M$ and $M'$ are adjacent under $\Rule$, we have $|M\bigtriangleup M'|\leq 4$.
    In addition, since $M$ and $M'$ are distinct perfect matchings of $G$, 
    $M$ and $M'$ differ by at least two elements, thus we have $|M\bigtriangleup M'|=4$.
    Consequently, \( M \bigtriangleup M' \) induces a 4-cycle, and thus the two matchings \( M \) and \( M' \) are adjacent under the flip operation.  
    Therefore, the sequence \( M_0, M_1, \ldots, M_\ell \) is a sequence between \( M_s \) and \( M_t \) such that any consecutive pair \( (M_i, M_{i+1}) \) is adjacent under the flip operation.  
    Hence, \( (G, M_s, M_t) \) is a yes-instance of \prb{Perfect Matching Reconfiguration}.
    
    This completes our proof.
\end{proof}

\section{PSPACE-completeness When $k$ is Superconstant}\label{sec:large_k}
This section is devoted to establishing the $\PSPACE$-completeness of \prb{ISR} and \prb{VCR} under $\kTJ$ when $k$ is superconstant in the initial independent set size $|I|$ and the initial vertex cover size $|S|$, respectively.

\subsection{ISR}\label{ISR_large_k}
The main result in this section is the following.
\begin{theorem}\label{ISR_PSPACEcom_largek}
    There exists some constant $\varepsilon_0\in(0,1)$ such that \prb{ISR} under $\kTJ$ on graphs of maximum degree $3$ is $\PSPACE$-complete for any $k$ satisfying the following condition:
    there exists a constant $c$ such that $k\leq \varepsilon_0|I|$ holds whenever $|I|\geq c$, where $I$ is the initial independent set of the input graph. 
\end{theorem}

To prove \Cref{ISR_PSPACEcom_largek}, we will construct a polynomial-time reduction from the \emph{optimization variant} of \prb{ISR} called \prb{Maxmin Independent Set Reconfiguration} (\prb{MaxminISR} for short)~\cite{IDHPSUU11}.
In the problem, we adopt the \emph{token addition and removal} ($\TAR$ for short)~\cite{IDHPSUU11}, under which two independent sets of a graph $G$ are adjacent if one is obtained from the other by adding or removing a single vertex of $G$.
In \prb{MaxminISR}, given a graph \( G \) and two independent sets \( I \), \( J \) of \( G \), we are asked to find a reconfiguration sequence $\sigma=\langle I=I_0, I_1, \ldots, I_\ell = J\rangle$ under \( \TAR \) that maximizes $\mathsf{val}(\sigma)$, where $\mathsf{val}(\sigma) = \min \{ |I_i|\colon 0\leq i \leq \ell \}$. 
For a graph $G$, we use $\alpha(G)$ to denote the number of vertices in a maximum independent set of $G$.
Let $(G, I, J)$ be an instance of \prb{MaxminISR}, and $\mathsf{val_{max}}(I,J)$ be the maximum value of $\mathsf{val}(\sigma)$ over all possible reconfiguration sequences $\sigma$ from $I$ to $J$ under $\TAR$.
Recently, the following \Cref{MaxminISRPSPACE-h} was proven~\cite{DBLP:conf/stoc/HiraharaO24,karthikc.s.2023inapproximability,DBLP:conf/stacs/Ohsaka23}.
\begin{theorem}[\cite{DBLP:conf/stoc/HiraharaO24,karthikc.s.2023inapproximability,DBLP:conf/stacs/Ohsaka23}]\label{MaxminISRPSPACE-h}
    Let $I$ and $J$ be initial and target independent sets of an input graph $G$ in \prb{MaxminISR}.
    Then, there exists some constant $\varepsilon_0\in (0,1)$ such that it is $\PSPACE$-hard to distinguish between the following two cases:
    \begin{description}[leftmargin=1em, nosep, topsep=0.5em]
        \item[{(i)}] $\mathsf{val_{max}}(I,J)\geq \alpha(G)-1$, and 
        \item[{(ii)}] $\mathsf{val_{max}}(I,J)<(1-\varepsilon_0)(\alpha(G)-1)$.
    \end{description}
    The same hardness result holds even when the maximum degree of $G$ is $3$, $|I|=|J|=\alpha(G)$, and $\frac{\alpha(G)}{|V(G)|}\in[\frac{1}{3},\frac{1}{2}]$.
\end{theorem}

To lead to \Cref{ISR_PSPACEcom_largek} from \Cref{MaxminISRPSPACE-h}, we provide the following lemma.

\begin{lemma}\label{MaxmintoISR}
    Let $I$ and $J$ be initial and target independent sets of an input graph $G$ in \prb{MaxminISR} and \prb{ISR}.
    Let \( f \) be a given function and \( g \) be any function defined on integers such that $x-g(x)\geq 1$ for all positive integers $x$ and there exists a fixed positive integer \( n_0 \) satisfying \( g(n) \geq f(n) \) for all integers \( n \geq n_0 \).
    Suppose that it is $\PSPACE$-hard to distinguish between the following two cases:
    \begin{description}[leftmargin=1em, nosep, topsep=0.5em]
        \item[{(i)}] $\mathsf{val_{max}}(I,J)\geq |I|-1$, and 
        \item[{(ii)}] $\mathsf{val_{max}}(I,J)< f(|I|)$.
    \end{description}
    Then, \prb{ISR} under $\kTJ$ is $\PSPACE$-hard, where $k=|I|-g(|I|)\geq 1$.
\end{lemma}
\begin{proof}
    Let $(G, I, J)$ be an instance of \prb{MaxminISR}, 
    and $(G, I, J, \kTJ)$ be an instance of \prb{ISR} where $k=|I|-g(|I|)$.
    We assume that $|I|\geq n_0$ and hence $k\ge 1$ since we can solve all instances with $|I|<n_0$ by enumerating all independent sets of constant size in polynomial time.
    We now show that if $(G, I, J)$ satisfies condition (i), then $(G, I, J, \kTJ)$ is a yes-instance, and if $(G, I, J)$ satisfies condition (ii), then $(G, I, J, \kTJ)$ is a no-instance.
    We will show the latter one by proving the contrapositive: if $(G, I, J, \kTJ)$ is a yes-instance, then $(G, I, J)$ does not satisfy condition (ii).

    Firstly, suppose that there is a reconfiguration sequence $\sigma=\langle I=I_0,I_1,\ldots,I_\ell = J\rangle$ between $I$ and $J$ under $\TAR$ such that $\mathsf{val}(\sigma)\geq |I|-1$. 
    It is known that this assumption holds if and only if there is a reconfiguration sequence under $\TJ$ between $I$ and $J$~\cite{KMM12}.
    Thus, there is a reconfiguration sequence under $\TJ$ between $I$ and $J$, and that is also a reconfiguration sequence under $\kTJ$. 
    Therefore, $(G, I, J, \kTJ)$ is a yes-instance.

    Conversely, suppose that there is a reconfiguration sequence $\sigma'=\langle I=I_0,I_1,\ldots,I_\ell = J\rangle$ between $I$ and $J$ under $\kTJ$ where $k=|I|-g(|I|)\geq 1$. 
    Then, for any two consecutive independent sets $I_{i-1}$ and $I_i$ with $i\in[\ell]$, we have $|I_{i-1}\cap I_i| \geq |I|-k = g(|I|)$.
    Additionally, we can transform from $I_{i-1}$ to $I_i$ under $\TAR$ as follows:
    Firstly, we remove tokens on vertices in $I_{i-1}\setminus I_i$ one by one; then, we add tokens on vertices in $I_i\setminus I_{i-1}$ one by one. 
    Through these steps, we have no independent set with size smaller than $|I_{i-1}\cap I_i|\geq g(|I|) \geq f(|I|)$.
    Therefore, there is a sequence $\sigma$ under $\TAR$ such that $\mathsf{val}(\sigma)\geq g(|I|)\geq f(|I|)$. That is, $(G, I, J)$ does not satisfy condition (ii).
    This completes the proof.
\end{proof}

We set $f(x)=(1-\varepsilon_0)(x-1)$ and let $g(x)$ be an arbitrary function such that $x-1\geq g(x)\geq f(x)$ for all $x\geq x_0$, for some constant $x_0$.
For example, $g(x)$ may be chosen as $x-c$ for some constant $c$,  $ x-\lceil\log x\rceil$, $ x-\lceil x^{1/2} \rceil$, or $x-\lceil \varepsilon x\rceil $ for some constant $\varepsilon \leq \varepsilon_0$.
Combining \Cref{MaxminISRPSPACE-h} and \Cref{MaxmintoISR}, \prb{ISR} under $\kTJ$ on graphs of maximum degree $3$ is $\PSPACE$-complete for $k=|I|-g(|I|)$, as claimed in \Cref{ISR_PSPACEcom_largek}.
This result includes the \(\PSPACE\)-completeness of \prb{ISR} under \(\kTJ\) for various values of \(k\), such as \(\Theta(1)\), \(\Theta(\log |I|)\), \(\Theta(|I|^\epsilon) \) with any constant $\epsilon\in(0,1)$, and \(\Theta(|I|)\).

\subsection{VCR}\label{VCR_large_k}
Similarly to \Cref{ISR_PSPACEcom_largek}, we prove the following \Cref{VCR_PSPACEcom_largek}.

\begin{theorem}\label{VCR_PSPACEcom_largek}
    There exists some constant $\varepsilon_0\in(0,1)$ such that \prb{VCR} under $\kTJ$ on graphs of maximum degree $3$ is $\PSPACE$-complete for any $k$ satisfying the following condition:
    there exists a constant $c$ such that $k\leq \varepsilon_0|S|$ holds whenever $|S|\geq c$, where $S$ is the initial vertex cover of the input graph.
\end{theorem}

To prove \Cref{VCR_PSPACEcom_largek}, we will construct a polynomial-time reduction from the \emph{optimization variant} of \prb{VCR} called \prb{Minmax Vertex Cover Reconfiguration} (\prb{MinmaxVCR} for short)~\cite{IDHPSUU11}.
In the problem, we adopt the token addition and removal ($\TAR$ for short) rule under which two vertex covers of a graph $G$ are adjacent if one is obtained from the other by adding or removing a single vertex of $G$.
In \prb{MinmaxVCR}, given a graph \( G \) and two vertex covers \( S \), \( T \) of \( G \), we are asked to find a reconfiguration sequence $\sigma=\langle S=S_0, S_1, \ldots, S_\ell = T\rangle$ under \( \TAR \) that minimizes $\mathsf{val}(\sigma)$, where $\mathsf{val}(\sigma) = \max \{ |S_i|\colon 0\leq i \leq \ell \}$. 
We use $\beta(G)$ to denote the size of a minimum vertex cover of a graph $G$.
Let $(G, S, T)$ be an instance of \prb{MinmaxVCR}, and $\mathsf{val_{min}}(S,T)$ be the minimum value of $\mathsf{val}(\sigma)$ over all possible reconfiguration sequences $\sigma$ from $S$ to $T$ under $\TAR$.
The following \Cref{MinmaxVCRPSPACE-h} was proved~\cite{DBLP:conf/stoc/HiraharaO24,karthikc.s.2023inapproximability,DBLP:conf/stacs/Ohsaka23}.
\begin{theorem}[\cite{DBLP:conf/stoc/HiraharaO24,karthikc.s.2023inapproximability,DBLP:conf/stacs/Ohsaka23}]\label{MinmaxVCRPSPACE-h}
    Let $S$ and $T$ be initial and target vertex covers of an input graph $G$ in \prb{MinmaxVCR}.
    Then, there exists some constant $\varepsilon_0\in (0,1)$ such that it is $\PSPACE$-hard to distinguish between the following two cases:
    \begin{description}[leftmargin=1em, nosep, topsep=0.5em]
        \item[{(i)}] $\mathsf{val_{min}}(S,T)\leq \beta(G)+1$, and 
        \item[{(ii)}] $\mathsf{val_{min}}(S,T)>(1+\varepsilon_0)(\beta(G)+1)$.
    \end{description}
    The same hardness result holds even when the maximum degree of $G$ is $3$, $|S|=|T|=\beta(G)$, and $\frac{\beta(G)}{|V(G)|}\in[\frac{1}{2},\frac{2}{3}]$.
\end{theorem}

To lead \Cref{VCR_PSPACEcom_largek} from \Cref{MinmaxVCRPSPACE-h}, we provide the following lemma.

\begin{lemma}\label{MinmaxtoVCR}
    Let $S$ and $T$ be initial and target vertex covers of an input graph $G$ in \prb{MinmaxVCR} and \prb{VCR}.
    Let \( f \) be a given function and \( g \) be any function defined on integers such that $g(x)\geq x+1$ for all positive integers $x$ and there exists a fixed positive integer \( n_0 \) satisfying \( g(n) \leq f(n) \) for all \( n \geq n_0 \).
    Suppose that it is $\PSPACE$-hard to distinguish between the following two cases:
    \begin{description}[leftmargin=1em, nosep, topsep=0.5em]
        \item[{(i)}] $\mathsf{val_{min}}(S,T)\leq |S|+1$, and 
        \item[{(ii)}] $\mathsf{val_{min}}(S,T)> f(|S|)$.
    \end{description}
    Then, \prb{VCR} under $\kTJ$ is $\PSPACE$-hard, where $k=g(|S|)-|S|\geq 1$.
\end{lemma}

\begin{proof}
    Let $(G,S,T)$ be an instance of \prb{MinmaxVCR}, and $(G,S,T,\kTJ)$ be an instance of \prb{VCR} where $k=g(|S|)-|S|$. We assume that $|S|\geq n_0$ and hence $k\geq 1$, since we can solve all instances with $|S|<n_0$ by enumerating all vertex covers of constant size in polynomial time. 
    We now show that if $(G,S,T)$ satisfies the condition (i), then $(G,S,T,\kTJ)$ is a yes-instance, and if $(G,S,T)$ satisfies the condition (ii), then $(G,S,T,\kTJ)$ is a no-instance. 
    We will show the latter one by proving the contrapositive: if $(G,S,T,\kTJ)$ is a yes-instance, then $(G,S,T)$ does not satisfy condition (ii).

    Firstly, assume that there is a reconfiguration sequence $\sigma=\langle S=S_0,S_1,\ldots,S_\ell = T\rangle$ between $S$ and $T$ under $\TAR$ such that $\mathsf{val}(\sigma)\leq |S|+1$. 
    Then, we observe that there is a reconfiguration sequence under $\TJ$ between $S$ and $T$~\cite{KMM12}.
    Moreover, that is also a reconfiguration sequence under $\kTJ$. 
    Therefore, $(G, S, T, \kTJ)$ is a yes-instance.

    Conversely, assume that there is a reconfiguration sequence $\sigma'=\langle S=S'_0,S'_1,\ldots,S'_\ell = T\rangle$ between $S$ and $T$ under $\kTJ$ where $k=g(|S|)-|S|\geq 1$. 
    Then, for any two consecutive vertex covers $S_{i-1}$ and $S_i$ with $i\in[\ell]$, we have $|S_{i-1}\cup S_i| \leq |S|+k = g(|S|)$.
    Additionally, we can transform from $S_{i-1}$ to $S_i$ under $\TAR$ as follows:
    Firstly, we add tokens on vertices in $S_i\setminus S_{i-1}$ one by one;
    then, we remove tokens on vertices in $S_{i-1}\setminus S_i$ one by one. 
    Through these steps, we have no vertex cover with size larger than $|S_{i-1}\cup S_i|\leq g(|S|) \leq f(|S|)$.
    Therefore, there is a sequence $\sigma$ under $\TAR$ such that $\mathsf{val}(\sigma) \leq f(|S|)$. That is, $(G, S, T)$ does not satisfy condition (ii).
    This completes the proof.
\end{proof}

We set $f(x)=(1+\varepsilon_0)(x+1)$ and let $g(x)$ be an arbitrary function such that $ g(x)\leq f(x)$ for all $x\geq x_0$, for some constant $x_0$.
For example, $g(x)$ may be chosen as $x+c$ for some constant $c$,  $ x+\lceil\log x\rceil$, $ x+\lceil x^{1/2} \rceil$, or $x+\lceil \varepsilon x\rceil $ for some positive constant $\varepsilon \leq \varepsilon_0$.
Combining \Cref{MinmaxVCRPSPACE-h} and \Cref{MinmaxtoVCR}, \prb{VCR} under $\kTJ$ on graphs of maximum degree $3$ is $\PSPACE$-complete for $k=g(|S|)-|S|$, as claimed in \Cref{VCR_PSPACEcom_largek}.
This result includes the \(\PSPACE\)-completeness of \prb{VCR} under \(\kTJ\) for various values of \(k\), such as \(\Theta(1)\), \(\Theta(\log |S|)\), \(\Theta(|S|^\epsilon )\) with any constant $\epsilon\in (0,1)$, and \(\Theta(|S|)\).

\section{Conclusion and Future Work}
In this paper, we investigated the computational complexity of the fundamental reconfiguration problems \prb{ISR} and \prb{VCR} on various graph classes under the extended reconfiguration rules $\kTJ$ and $\kTS$.

The following open problems are suggested for future research:
(1)~Is \prb{ISR} under $\kTJ$ with \( k = |I| - \mu \) $\NP$-hard on planar graphs of maximum degree 3 for any fixed positive integer \( \mu \)?  
(2)~Is \prb{ISR} under $\kTJ$ with \( k = |I| - \mu \) in $\NP$ on general graphs not only when $\mu$ is fixed but also when \( \mu = O(\log |I|) \)?
(See also \Cref{ISRresultskTJ}.)

\bibliographystyle{plainurl}

 
\end{document}